\documentclass[11pt]{article} %%{proc}
\usepackage{latexsym}
\usepackage{lscape}
\usepackage{verbatim} %%for begin{comment} .. end{comment}

\usepackage[pdftex]{graphicx}   %% Needed for including pdf-figures.

\normalsize
\makeatletter
\@addtoreset{equation}{section}
\makeatother

\newtheorem{theorem}{Theorem}[section]
\newtheorem{lemma}[theorem]{Lemma}

\newenvironment{proof}{{\sc Proof. }}{\hfill$\Box$\vspace{0.1in}}
\title{An Approximation Algorithm for Maximum Internal Spanning Tree}
\author{
Zhi-Zhong Chen\thanks{ Corresponding author. 
Division of Information System Design, 
Tokyo Denki University, 
Hatoyama, Saitama 350-0394, Japan. 
Email: \mbox{{\tt zzchen@mail.dendai.ac.jp}}}
\and 
Youta Harada\thanks{
Division of Information System Design, 
Tokyo Denki University, 
Hatoyama, Saitama 350-0394, Japan.}
\and
Lusheng Wang\thanks{Department of Computer Science,
City University of Hong Kong, 
Tat Chee Avenue, Kowloon, Hong Kong SAR.
Email: {\tt cswangl@cityu.edu.hk}}
}
\date{}
\begin{document}
\maketitle
%==============================================================================
\begin{abstract}
Given a graph $G$, the {\em maximum internal spanning tree problem} (MIST for short) 
asks for computing a spanning tree $T$ of $G$ such that the number of internal vertices 
in $T$ is maximized. MIST has possible applications in the design of cost-efficient 
communication networks and water supply networks and hence has been extensively studied 
in the literature. MIST is NP-hard and hence a number of 
polynomial-time approximation algorithms have been designed for MIST in the literature. 
The previously best polynomial-time approximation algorithm for MIST achieves a ratio of 
$\frac{3}{4}$. In this paper, we first design a simpler algorithm that achieves the same 
ratio and the same time complexity as the previous best. We then refine the algorithm into 
a new approximation algorithm that achieves a better ratio (namely, $\frac{13}{17}$) with 
the same time complexity. Our new algorithm explores much deeper structure of the problem 
than the previous best. The discovered structure may be used to design 
even better approximation or parameterized algorithms for the problem in the future. 
\end{abstract}
%==============================================================================


\section{Introduction}\label{sec:intro}
%========================================================================== 
The {\em maximum internal spanning tree problem} (MIST for short) requires the computation of 
a spanning tree $T$ in a given graph $G$ such that the number of internal vertices in $T$ is 
maximized. MIST has possible applications in the design of cost-efficient communication networks 
\cite{SW08} and water supply networks \cite{BFGL13}. Unfortunately, MIST is clearly NP-hard 
because the problem of finding a Hamiltonian path in a given graph is NP-hard \cite{GJ79} and 
can be easily reduced to MIST. MIST is in fact APX-hard \cite{LZW14} and hence does not admit 
a polynomial-time approximation scheme. 

Since MIST is APX-hard, it is of interest to design polynomial-time approximation algorithms 
for it that achieve a constant ratio as close to~1 as possible. 
Indeed, Prieto and Sliper \cite{PS03} presented a polynomial-time approximation algorithm 
for MIST achieving a ratio of $\frac{1}{2}$. Their algorithm is based on local search. 
By slightly modifying Prieto and Sliper's algorithm, Salamon and Wiener \cite{SW08} then obtained 
a faster (linear-time) approximation algorithm achieving the same ratio. Salamon and Wiener 
\cite{SW08} also considered two special cases of MISP. More specifically, they \cite{SW08} designed 
a polynomial-time approximation algorithm for the special case of MIST restricted to claw-free 
graphs that achieves a ratio of $\frac{2}{3}$, and also designed a polynomial-time approximation 
algorithm for the special case of MIST restricted to cubic graphs that achieves a ratio of 
$\frac{5}{6}$. Salamon \cite{Sal09} later proved that the approximation algorithm in \cite{SW08} 
indeed achieves a performance ratio of $\frac{3}{r+1}$ for the special case of MIST restricted 
to $r$-regular graphs. Based on local optimization, Salamon \cite{Sal09b} further came up with an 
$O(n^4)$-time approximation algorithm for the special of MIST restricted to graphs without leaves 
that achieves a ratio of $\frac{4}{7}$. The algorithm in \cite{Sal09b} was subsequently simplified 
and re-analyzed by Knauer and Spoerhase \cite{KJ09} so that it runs faster (in cubic time) and 
achieves a better ratio (namely, $\frac{3}{5}$) for (the general) MIST. Li {\em et al.} 
\cite{LCW14} even went further by showing that a deeper local search than those in \cite{KJ09} 
and \cite{Sal09b} can achieve a ratio of $\frac{2}{3}$ for MIST. Recently, Li and Zhu~\cite{LZW14} 
came up with a polynomial-time approximation algorithm for MIST that achieves a ratio of 
$\frac{3}{4}$. Unlike the other previously known approximation algorithms for MIST, the algorithm 
in \cite{LZW14} is based on a simple but crucial observation that the maximum number of internal 
vertices in a spanning tree of a graph $G$ can be bounded from above by the maximum number of edges 
in a triangle-free path-cycle cover of $G$. 

In the weighted version of MIST (WMIST for short), each vertex of the given graph $G$ has 
a nonnegative weight and the objective is to find a spanning tree $T$ of $G$ such that 
the total weight of internal vertices in $T$ is maximized. Salamon \cite{Sal09b} designed 
an $O(n^4)$-time approximation for WMIST that achieves a ratio of $\frac{1}{2\Delta-3}$, 
where $\Delta$ is the maximum degree of a vertex in the input graph. Salamon \cite{Sal09b} 
also considered the special case of WMIST restricted to claw-free graphs without leaves, 
and designed an $O(n^4)$-time approximation algorithm for the special case that achieves 
a ratio of $\frac{1}{2}$. Subsequently, Knauer and Spoerhase \cite{KJ09} proposed a 
polynomial-time approximation algorithm for (the general) WMIST that achieves a ratio 
of $\frac{1}{3} - \epsilon$ for any constant $\epsilon > 0$. 

In the parameterized version of MIST (PMIST for short), we are asked to decide whether 
a given graph $G$ has a spanning tree with at least a given number $k$ of internal vertices. 
PMIST and its special cases and variants have also been extensively studied in the literature 
\cite{BFGL13,CRG+10,FLGS12,FGST13,LJF16,LWCC15,PS03,Pre03,PS05}. The best known kernel for 
PMIST is of size $2k$ and it leads to the fastest known algorithm for PMIST with running time 
$O(4^kn^{O(1)})$ \cite{LWCC15}.

In this paper, we first give a new approximation algorithm for MIST that is simpler than 
the one in \cite{LZW14} but achieves the same approximation ratio and time complexity. 
In more details, the time complexity is dominated by that of computing a maximum triangle-free 
path-cycle cover in a graph.  We then show that the algorithm can be refined into a new 
approximation algorithm for MIST that has the same time complexity as the algorithm 
in~\cite{LZW14} but achieves a better ratio (namely, $\frac{13}{17}$). 
To obtain our algorithm, we use three new main ideas. The first main idea 
is to bound the maximum number of internal vertices in a spanning tree of a graph $G$ by the 
maximum number of edges in a {\em special} (rather than general) triangle-free path-cycle cover 
of $G$. Roughly speaking, we can figure out that certain vertices in $G$ must be leaves 
in an optimal spanning tree of $G$, and hence we can require that the degrees of these 
vertices be at most~1 when computing a maximum triangle-free path-cycle cover ${\cal C}$ 
of $G$. In this sense, ${\cal C}$ is special and can have significantly fewer edges than 
a maximum (general) triangle-free path-cycle cover of $G$, and hence gives us a tighter 
upper bound. The second idea is to carefully modify ${\cal C}$ into a spanning tree $T$ 
by local improvement. Unfortunately, we can not always guarantee that the number of 
internal vertices in $T$ is at least $\frac{13}{17}$ times the number of edges in ${\cal C}$. 
Our third idea is to show that if this unfortunate case occurs, then an optimal spanning 
tree of $G$ cannot have so many internal vertices. These ideas may be used to design 
even better approximation or parameterized algorithms for MIST in the future. 

The remainder of this paper is organized as follows. Section~\ref{sec:def} gives basic 
definitions that will be used in the remainder of the paper. Section~\ref{sec:simple} 
presents a simple approximation algorithm for MIST that achieves a ratio of $\frac{3}{4}$. 
The subsequent sections are devoted to refining the algorithm so that it achieves a better
ratio.

\section{Basic Definitions}\label{sec:def}
Throughout this chapter, a graph means a simple undirected graph (i.e., it has neither 
parallel edges nor self-loops). 

Let $G$ be a graph. We denote the vertex set of $G$ by $V(G)$, and denote the edge set of $G$ 
by $E(G)$. For a subset $U$ of $V(G)$, $G-U$ denotes the graph obtained from $G$ by removing 
the vertices in $U$ (together with the edges incident to them), while $G[U]$ denotes 
$G - (V(G)\setminus U)$. We call $G[U]$ the {\em subgraph of $G$ induced by $U$}. 
For a subset $F$ of $E(G)$, $G-F$ denotes the graph obtained from $G$ by removing the edges in $F$. 
An edge $e$ of $G$ is a {\em bridge} of $G$ if 
$G - \{e\}$ has more connected components than $G$, and is a {\em non-bridge} otherwise. 
A vertex $v$ of $G$ is a {\em cut-point} if $G - \{v\}$ has more connected components than $G$. 

Let $v$ be a vertex of $G$. The {\em neighborhood} of $v$ in $G$, denoted by $N_G(v)$, 
is $\{u~|~\{v,u\}\in E(G)\}$. The {\em degree} of $v$ in $G$, denoted by $d_G(v)$, is 
$|N_G(v)|$. If $d_G(v) = 0$, then $v$ is an {\em isolated} vertex of $G$. 
If $d_G(v) \le 1$, then $v$ is a {\em leaf} of $G$; otherwise, $v$ is a {\em 
non-leaf} of $G$. We use $L(G)$ to denote the set of leaves in $G$. 

Let $H$ be a subgraph of $G$. $N_G(H)$ denotes $\bigcup_{v\in V(H)} N_G(v) \setminus V(H)$. 
A {\em port} of $H$ is a $u \in V(H)$ with $N_G(u) \setminus V(H) \ne \emptyset$. 
When $H$ is a path, $H$ is {\em dead} if neither endpoint of $H$ is a port of $H$, 
while $H$ is {\em alive} otherwise. 
%$H$ and $v$ are {\em adjacent} in $G$ if $v \not\in V(H)$ but $v \in N_G(H)$. 
$H$ and another subgraph $H'$ of $G$ are {\em adjacent} in $G$ if $V(H) \cap V(H') = \emptyset$ 
but $N_G(H) \cap V(H') \ne \emptyset$ (or equivalently, $N_G(H') \cap V(H) \ne \emptyset$). 

A {\em cycle} in $G$ is a connected subgraph of $G$ in which each vertex is of degree~2. 
A {\em path} in $G$ is either a single vertex of $G$ or a connected subgraph of $G$ in which 
exactly two vertices are of degree~1 and the others are of degree~2. A vertex $v$ of a path 
$P$ in $G$ is an {\em endpoint} of $P$ if $d_P(v) \le 1$, and is an {\em internal vertex} of 
$P$ if $d_P(v) = 2$. The {\em length} of a cycle or path $C$ is the number of edges in $C$ 
and is denoted by $|C|$. A {\em $k$-cycle} is a cycle of length~$k$, while a {\em $k$-path} 
is a path of length $k$. A {\em tree} (respectively, {\em cycle}) {\em component} of $G$ is 
a connected component of $G$ that is a tree (respectively, cycle). In particular, if a tree 
component $T$ of $G$ is indeed a path (respectively, $k$-path), then we call $T$ a {\em path} 
(respectively, {\em $k$-path}) {\em component} of $G$. 

A {\em tree-cycle cover} (TCC for short) of $G$ is a subgraph $H$ of $G$ such that $V(H) 
= V(G)$ and each connected component of $H$ is a tree or cycle. Let $H$ be a TCC of $G$. 
$H$ is a {\em Hamiltonian path} (respectively, {\em cycle}) of $G$ if $H$ is a path 
(respectively, cycle), and is a {\em spanning tree} of $G$ if $H$ is a tree. 
$H$ is a {\em path-cycle cover} (PCC for short) of $G$ if each tree component of $H$ is 
a path. $H$ is a {\em path cover} of $G$ if $H$ has only path components. 
A {\em triangle-free} TCC (TFTCC for short) of $G$ is a TCC without 3-cycles. Similarly, 
a {\em triangle-free} PCC (TFPCC for short) of $G$ is a PCC without 3-cycles. A TFPCC of $G$ 
is {\em maximum} if its number of edges is maximized over all TFPCCs of $G$. For convenience, 
let $t(n,m)$ denote the time complexity of computing a maximum TFPCC in a graph with 
$n$ vertices and $m$ edges. It is known that $t(n,m) = O(n^2m^2)$~\cite{Har84}. 

Suppose that $G$ is connected. The {\em weight} of a spanning tree $T$ of $G$, denoted by $w(T)$, 
is the number of non-leaves in $T$. We use $opt(G)$ to denote the maximum weight of a spanning 
tree of $G$. An {\em optimal spanning tree} (OST for short) of $G$ is a spanning tree $T$ of $G$ 
with $w(T) = opt(G)$.

\section{A Simple 0.75-Approximation Algorithm}\label{sec:simple}
Throughout the remainder of this paper, $G$ means a connected graph for which we want to find 
an OST. Moreover, $T$ denotes an OST of $G$. For convenience, let $n = |V(G)|$ and $m = |E(G)|$.

\subsection{Reduction Rules}\label{subsec:rule}
We want to make $G$ smaller (say, by deleting one or more vertices or edges from $G$) 
without decreasing $opt(G)$. For this purpose, we define two {\em strongly safe} operations 
on $G$ below. Here, an operation on $G$ is {\em strongly safe} if performing it on $G$ does not 
change $opt(G)$. 

\begin{description}
\item[Operation 1.] If $|V(G)| > 3$ and $E(G)$ contains two edges $\{u_1,v\}$ and $\{u_2,v\}$ 
	such that both $u_1$ and $u_2$ are leaves of $G$, then delete $u_2$. 
\item[Operation 2.] If for a non-bridge $e = \{u_1,u_2\}$ of $G$, $G - \{u_i\}$ has a connected 
	component $K_i$ with $u_{3-i} \not\in V(K_i)$ for each $i\in\{1,2\}$, then delete $e$. 
	({\em Comment:} When $|V(K_1)| = |V(K_2)| = 1$, 
	Li and Zhu~\cite{LZW14} showed that Operation~2 is strongly safe.) 
\end{description} 

\begin{center}
\begin{figure}[htpb]
\centerline{\includegraphics[scale=.9]{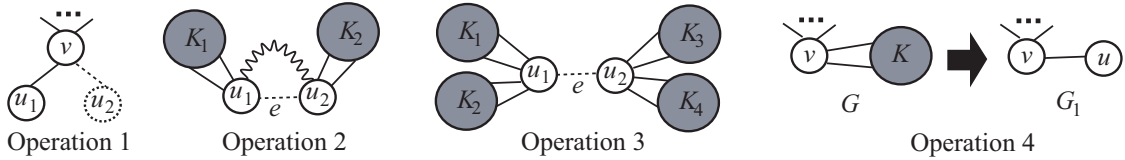}}
\caption{Operations 1 through 4, where the wavy curve is a path and each dotted edge and 
the vertex enclosed by a dotted circle will be deleted.}
\label{fig:op1-4}
\end{figure}
\end{center}

\begin{lemma}\label{lem:op1}{\rm \cite{LZW14}}
Operation 1 is strongly safe.
\end{lemma}

\begin{lemma}\label{lem:op2}
Operation 2 is strongly safe.
\end{lemma}
\begin{proof}
If $e \not\in E(T)$, we are done. So, assume that $e \in E(T)$. Obviously, at least one vertex 
$v_1$ of $K_1$ is adjacent to $u_1$ in $T$ because $T$ is connected. So, $\{u_2,v_1\} \subseteq 
N_T(u_1)$. Similarly, $\{u_1,v_2\} \subseteq N_T(u_2)$ for some vertex $v_2$ of $K_2$. Moreover, 
since $e$ is a non-bridge of $G$, $G - \{u_1,u_2\}$ has a connected component $K_3$ (other than 
$K_1$ and $K_2$) with $\{u_1,u_2\} \subseteq N_G(K_3)$. Since $T$ is connected, $u_1$ or $u_2$ 
is adjacent to a vertex $v_3$ of $K_3$ in $T$. We assume that $v_3 \in N_T(u_1)$; the other 
case is similar. Then, after deleting $e$ from $T$, only $u_2$ may become a new leaf. If $u_2$ 
becomes a leaf in $T - \{e\}$, then all vertices of $K_3$ must belong to the component tree of 
$T - \{e\}$ containing $u_1$ and hence adding an arbitrary edge $\{u_2,v_4\}$ of $G$ with 
$v_4 \in V(K_3)$ to $T - \{e\}$ yields a new OST of $G$. So, we may assume that $u_2$ does not 
become a leaf in $T - \{e\}$. Then, since $e$ is a non-bridge of $G$, $G$ must have an edge 
$\{x_1,x_2\}$ such that for each $i \in \{1,2\}$, $x_i$ belongs to the component tree of 
$T - \{e\}$ containing $u_i$. Now, adding the edge $\{x_1,x_2\}$ to $T - \{e\}$ yields 
a new OST of $G$. 
\end{proof}

An operation on $G$ is {\em weakly safe} if performing it on $G$ yields one or more graphs 
$G_1$, \ldots, $G_k$ such that (1)~$|V(G)| \ge \sum_{i=1}^k |V(G_i)|$, 
$|E(G)| \ge \sum_{i=1}^k |E(G_i)|$, and $|V(G)| + |E(G)| > \sum_{i=1}^k |V(G_i)| 
+ \sum_{i=1}^k |E(G_i)|$, (2)~$opt(G) \le \sum_{i=1}^k opt(G_i) + c$ for some nonnegative 
integer $c$, and (3)~given a spanning tree $T_i$ for each $G_i$, a spanning tree $T$ of 
$G$ with $w(T) \ge \sum_{i=1}^k w(T_i) + c$ can be computed in linear time. 
Note that the last two conditions in the definition imply that 
$opt(G) = \sum_{i=1}^k opt(G_i) + c$. 

\begin{description}
\item[Operation 3.] If $G$ has a bridge $e = \{u_1,u_2\}$ such that for each $i\in\{1,2\}$, 
	$u_i$ is a cut-point in the connected component $G_i$ of $G - e$ with $u_i \in V(G_i)$, 
	then obtain $G_1$ and $G_2$ as the connected components of $G - e$. 
\item[Operation 4.] If $G$ has a cut-point $v$ such that one connected component $K$ of 
	$G - \{v\}$ has at least two but at most 8 vertices, then obtain $G_1$ from $G - V(K)$ 
	by adding a new vertex $u$ and a new edge $\{v,u\}$. 
\end{description} 

The number 8 in the definition of Operation~4 is not essential. It can be chosen at 
one's discretion as long as it is a constant. We here choose the number 8, because 
it will be the smallest number for the proofs of several lemmas in this paper to go through. 

\begin{lemma}\label{lem:op7}
Operation 3 is weakly safe.
\end{lemma}
\begin{proof}
First, we want to show that $opt(G) \le opt(G_1) + opt(G_2)$. 
Consider an $i \in \{1,2\}$. Since $u_i$ is a cut-point in $G_i$, 
$d_T(u_i) \ge 3$. Thus, the degree of $u_i$ in $T - \{e\}$ is at least~2. So, one 
component tree of $T - \{e\}$ is a spanning tree of $G_1$, the other is a spanning tree 
of $G_2$, and their total weights equals $w(T)$. Thus, $opt(G) \le opt(G_1) + opt(G_2)$.

Next, suppose that for each $i\in\{1,2\}$, $T_i$ is a spanning tree of $G_i$.
Since $u_i$ is a cut-point in $G_i$, $d_{T_i}(u_i) \ge 2$. So, using $e$ to connect $T_1$ 
and $T_2$ into a single tree yields a spanning tree of $G$ whose weight is $w(T_1)+w(T_2)$. 
\end{proof}

\begin{lemma}\label{lem:op8}
Operation 4 is weakly safe.
\end{lemma}
\begin{proof}
Let $K'$ be the graph obtained from $G[V(K)\cup\{v\}]$ by adding a new vertex $u'$ and 
a new edge $\{u',v\}$. Let 
$c = opt(K') - 1$. 

First, we want to show that $opt(G) \le opt(G_1) + c$. 
Since $v$ is a cut-point of $G$, $d_T(v) \ge 2$. Let $T_1$ be the spanning tree 
of $G_1$ obtained from $T - V(K)$ by adding $u$ and the edge $\{v,u\}$. Further let 
$T'$ be the spanning tree of $K'$ obtained from $T[\{v\} \cup V(K)]$ by adding $u'$ and 
edge $\{u',v\}$. Clearly, $w(T) = w(T_1) + w(T') - 1 \le opt(G_1) + c$. 
Thus, $opt(G) \le opt(G_1) + c$. 
 
Next, suppose that $T_1$ is a spanning tree of $G_1$. Let $T'$ be an OST of $K'$. We can 
obtain a spanning tree $\tilde{T}$ of $G$ from $T_1$ by first deleting $u$, next adding 
$T'[V(K)]$, and further adding new edges to connect $v$ to those vertices of $V(K)$ that 
are adjacent to $v$ in $T'$. Obviously, $u \in L(T_1)$, $u' \in L(T')$, $v \not\in L(T_1)$, 
$v \not\in L(T')$, the degree of each vertex $x$ of $T_1$ other than $v$ and $u$ in $\tilde{T}$ 
is $d_{T_1}(x)$, and the degree of each vertex $y$ of $T'$ other than $v$ and $u'$ in 
$\tilde{T}$ is $d_{T'}(y)$. Thus, $w(\tilde{T}) = w(T_1) + w(T') - 1 = w(T_1) + c$. 
\end{proof}

An operation on $G$ is {\em safe} if it is strongly or weakly safe on $G$.

\subsection{The Algorithm}\label{subsec:simpleAlgo}
As in \cite{LZW14}, the algorithm is based on a lemma which says that $G$ has a path cover 
${\cal P}$ such that $opt(G)$ is bounded from above by the number of edges in ${\cal P}$. 
We next state the lemma in a stronger form and give an extremely simple proof. 

\begin{lemma}{\rm \label{lem:LZW}}
Given a spanning tree $\tilde{T}$ of $G$, we can construct a path cover ${\cal P}$ of $G$ 
such that $|E({\cal P})| \ge w(\tilde{T})$ and $d_{\cal P}(v) \le 1$ for each leaf $v$ of 
$\tilde{T}$. 
\end{lemma}
\begin{proof}
We simply construct ${\cal P}$ from $\tilde{T}$ by first rooting $\tilde{T}$ at an arbitrary 
non-leaf and then for each non-leaf $u$ of $\tilde{T}$, deleting all but one edge between $u$ 
and its children. 
\end{proof}

Now, the outline of the algorithm is as follows.

\begin{enumerate}
\item\label{step:ssafe}
	Whenever there is an $i\in\{1,2\}$ such that Operation~$i$ can be performed 
	on $G$, then perform Operation $i$ on $G$. 
\item\label{step:wsafe} 
	Whenever there is an $i\in\{3,4\}$ such that Operation~$i$ can be performed 
	on $G$, then perform the following steps:
  \begin{enumerate}
  \item Perform Operation $i$ on $G$. Let $G_1$, \ldots, $G_k$ be the resulting graphs. 
  \item\label{substep:recur} 
	For each $j\in\{1,\ldots,k\}$, compute a spanning tree $T_j$ of $G_j$ recursively.
  \item Combine $T_1$, \ldots, $T_k$ into a spanning tree $\tilde{T}$ of $G$ such that 
	$w(\tilde{T}) \ge \sum_{i=1}^k w(T_i) + c$. 
  \item Return $\tilde{T}$. 
  \end{enumerate} 
\item\label{step:small} 
	If $|V(G)| \le 8$, then compute and return an OST of $G$ in $O(1)$ time. 
\item\label{step:mtfpc} 
	Compute a maximum TFPCC ${\cal C}$ of $G$. 
	({\em Comment:} By Lemma~\ref{lem:LZW}, $opt(G) \le |E({\cal C})|$). 

\item\label{step:prep} 
	Perform a preprocessing on ${\cal C}$ without decreasing $|E({\cal C})|$. 
\item\label{step:trans} 
	Transform ${\cal C}$ into a spanning tree $\tilde{T}$ of $G$ such that 
	$w(\tilde{T}) \ge \frac{3}{4}|E({\cal C})|$. 
\item Return $\tilde{T}$. 
\end{enumerate}

Only Steps~\ref{step:prep} and~\ref{step:trans} are unclear. So, we detail them below.
First, Step~\ref{step:prep} is done by performing the next three operations until none 
of them is applicable. 

\begin{description}
\item[Operation 5.] 
	If ${\cal C}$ has a dead path component $P$ such that $2 \le |P| \le 4$ and 
	$G[V(P)]$ has an alive Hamiltonian path $Q$, then replace $P$ by $Q$. 
\item[Operation 6.] 
	If an endpoint $u$ of a path component $P$ of ${\cal C}$ is adjacent to a vertex $v$ 
	of a cycle $C$ of ${\cal C}$ in $G$, then combine $P$ and $C$ into a single path by 
	replacing one edge incident to $v$ in $C$ with the edge $\{u,v\}$. 
\item[Operation 7.]
	If an endpoint $u_1$ of a path component $P_1$ of ${\cal C}$ is adjacent to an internal 
	vertex $u_2$ of another path component $P_2$ in $G$ such that one edge $e'$ incident to 
	$u_2$ in $P_2$ satisfies that combining $P_1$ and $P_2$ by replacing $e'$ with the edge 
	$\{u_1,u_2\}$ yields two paths $Q_1$ and $Q_2$ with $\max\{|Q_1|, |Q_2|\} > 
	\max\{|P_1|, |P_2|\}$, then replace $P_1$ and $P_2$ by $Q_1$ and $Q_2$. 
	({\em Comment:} For each $i\in\{5,6,7\}$, Operation~$i$ does not change the maximality 
	of ${\cal C}$. So, due to the maximality of ${\cal C}$, no endpoint of a path component 
	$P_1$ of ${\cal C}$ is adjacent to an endpoint of another path component $P_2$ in $G$.)
\end{description}

\begin{center}
\begin{figure}[htpb]
\centerline{\includegraphics[scale=.9]{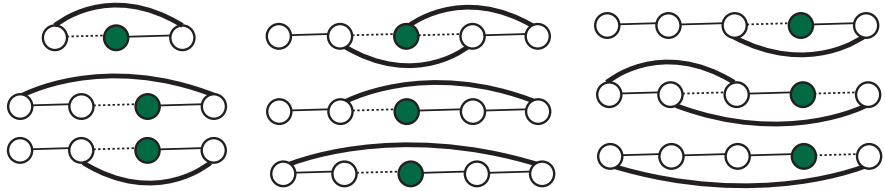}}
\caption{The possible cases of Operation 5, where each filled circle is a port, 
each dotted edge will be deleted, and each bold edge will be added.}
\label{fig:op5}
\end{figure}
\end{center}

\begin{center}
\begin{figure}[htpb]
\centerline{\includegraphics[scale=.9]{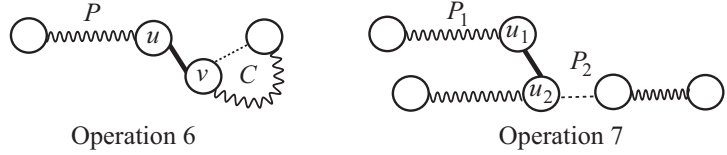}}
\caption{Operations 6 and 7, where each wavy line or curve is a path, 
each dotted edge will be deleted, and each bold edge will be added.}
\label{fig:op6-7}
\end{figure}
\end{center}

\begin{lemma}\label{lem:prep0}
Immediately after Step~\ref{step:prep}, the following statements hold:
\begin{enumerate}
\item\label{stat:q1}
	${\cal C}$ is a maximum TFPCC of $G$ and hence has at least $opt(G)$ edges. 
\item\label{stat:q3} 
	If a path component $P$ of ${\cal C}$ is of length at most~3, then $P$ is alive. 
\item\label{stat:q5} 
	If an endpoint $v$ of a path component $P$ of ${\cal C}$ is a port of $P$, then 
	each vertex in $N_G(v) \setminus V(P)$ is %a vertex of a cycle of ${\cal C}$ or 
	an internal vertex of a path component $Q$ of ${\cal C}$ with $|Q| \ge 2|P| + 2$.
\end{enumerate}
\end{lemma}
\begin{proof}
We prove the statements separately as follows.

{\em Statement~\ref{stat:q1}:} Immediately before Step~\ref{step:prep}, ${\cal C}$ has 
is a maximum TFPCC of $G$. Since Operations~5 through~7 keep ${\cal C}$ being 
a TFPCC without changing the number of edges in ${\cal C}$, Statement~\ref{stat:p1} holds. 

{\em Statement~\ref{stat:q3}:} Let $P$ be a path component of ${\cal C}$ with $|P|\le 3$. 
If $|P| \le 1$, then $P$ is alive because otherwise $G$ would be disconnected. 
So, $|P| = 2$ or 3. Let $u_1$ and $u_2$ be the endpoints of $P$. For a contradiction, 
assume that $P$ is dead. Then, since $G$ is connected, $P$ has at least 
one internal vertex $x$ adjacent to a vertex $x' \in V(G) \setminus V(P)$ in $G$. If 
$\{u_1,u_2\} \in E(G)$, then $G[V(P)]$ has a Hamiltonian path $Q$ in which $x$ is an 
endpoint, contradicting the fact that Operation~5 cannot be performed on ${\cal C}$. 
So, we assume that $\{u_1,u_2\} \not\in E(G)$. Now, if $|P| = 2$, then Operation~1 can 
be performed on $G$, a contradiction. Thus, we further assume that $|P| = 3$. 
Then, since Operation~4 cannot be performed on $G$, the other internal vertex $y$ 
(than $x$) of $P$ is adjacent to a vertex $y'\in V(G) \setminus V(P)$ in $G$. 
Now, if $G[V(P)]$ is not $P$ itself, then Operation~5 can be performed on ${\cal C}$, 
a contradiction; otherwise, Operation~2 or~3 can be performed on $G$, 
a contradiction. Note that it does not matter whether $x' = y'$ or not. 

{\em Statement~\ref{stat:q5}:} 
	Suppose that an endpoint $v$ of a path component $P$ of ${\cal C}$ is a port. 
	Consider an arbitrary $u \in N_G(v) \setminus V(P)$. Since Operation~6 is not 
	applicable on ${\cal C}$, $u$ appears in a path component $Q$ of ${\cal C}$. Then, 
	by the comment on Operation~7, $u$ is an internal vertex of $Q$. Let $u_1$ and $u_2$ be 
	the endpoints of $Q$. For each $i\in\{1,2\}$, let $Q_i$ be the path from $u$ to $u_i$ in $P$. 
	Then, $|Q| = |Q_1| + |Q_2|$. Moreover, since Operation~7 cannot be applied on ${\cal C}$, 
	$|P| + |Q_i| + 1 \le |Q|$ for each $i\in\{1,2\}$. Thus, $2|P| + 2 \le |Q|$.
\end{proof}

We next detail Step~\ref{step:trans}. First, for each path component $P$ of ${\cal C}$ with 
$1\le |P| \le 3$, we select one edge $e_P \in E(G)$ connecting an endpoint of $P$ to a vertex 
not in $P$, and add $e_P$ to an initially empty set $M$. 
Such $e_P$ exists by Statement~\ref{stat:q3} in Lemma~\ref{lem:prep0}. Moreover, 
by Statement~\ref{stat:q5} in Lemma~\ref{lem:prep0}, the endpoint of $e_P$ not in $P$ appears 
in a path component $Q$ of ${\cal C}$ with $|Q| \ge 4$. So, for two path components $P_1$ and 
$P_2$ in ${\cal C}$, $e_{P_1} \ne e_{P_2}$. Consider the graph $H$ obtained from ${\cal C}$ 
by adding the edges in $M$. Each connected component of $H$ is a cycle of length at least~4
or a tree. Suppose that we modify $H$ by performing the following three steps in turn:
\begin{itemize}
\item Whenever $H$ has two cycles $C_1$ and $C_2$ such that some edge $e=\{u_1,u_2\} \in E(G)$
satisfies $u_1\in V(C_1)$ and $u_2\in V(C_2)$, delete one edge of $C_1$ incident to $u_1$ 
from $H$, delete one edge of $C_2$ incident to $u_2$ from $H$, and add $e$ to $H$. 

\item Whenever $H$ has a cycle $C$, choose an edge $e=\{u,v\}\in E(G)$ with $u \in V(C)$ and 
$v\not\in V(C)$, delete one edge of $C$ incident to $u$ from $H$, and add $e$ to $H$. 

\item Whenever $H$ has two connected components $C_1$ and $C_2$ such that some edge 
$e=\{u_1,u_2\} \in E(G)$ satisfies $u_1\in V(C_1)$ and $u_2\in V(C_2)$, add $e$ to $H$. 
\end{itemize}

Step~\ref{step:trans} is done by obtaining $\tilde{T}$ as the final modified $H$. Obviously, 
for each cycle $C$ of ${\cal C}$, at least $|C| - 1\ge \frac{3}{4}|C|$ vertices of $C$ are 
internal vertices of $\tilde{T}$. Moreover, for each path component $P$ of ${\cal C}$ with 
$|P|\ge 4$, at least $|P| - 1 \ge \frac{3}{4}|P|$ vertices of $P$ are internal vertices of 
$\tilde{T}$. Furthermore, for each path component $P$ of ${\cal C}$ with $1\le |P|\le 3$, 
at least $|P|$ vertices of $P$ are internal vertices of $\tilde{T}$. So, $\tilde{T}$ has 
at least $\frac{3}{4}|E({\cal C})|$ internal vertices. Obviously, all steps of the algorithm 
excluding Steps~\ref{substep:recur} and~\ref{step:mtfpc} can be done in $O(|E(G)|^2)$ time. 
Now, we have the following theorem:

\begin{theorem}\label{th:main0}
The algorithm achieves an approximation ratio of $\frac{3}{4}$ and runs in $O(m^2) + t(n, m)$ time. 
\end{theorem}

In the sequel, we consider how to improve the algorithm. The first idea is to introduce 
more safe reduction rules (cf. Section~\ref{sec:reduce}). The second idea is to compute a 
better upper bound on $opt(G)$ than that given by a maximum TFPCC (cf. Section~\ref{sec:tfpcc}). 
The third idea is to perform a more sophisticated preprocessing on ${\cal C}$ 
(cf. Section~\ref{sec:prep}). The last idea is to transform ${\cal C}$ into a spanning tree 
of $G$ more carefully (cf. Section~\ref{sec:trans}).

\section{More Safe Reduction Rules}\label{sec:reduce}
In addition to the four safe reduction rules in Section~\ref{subsec:rule}, 
we further introduce the following rules. 

\begin{description}
\item[Operation 8.] If for four vertices $u_1$, \ldots, $u_4$, $N_G(u_3) = N_G(u_4) = \{u_1,u_2\}$, 
	$G - \{u_2\}$ has a connected component $K$ with $u_1\not\in V(K)$, then 
	delete the edge $e = \{u_2,u_3\}$. 
\item[Operation 9.] If for five vertices $u_1$, \ldots, $u_5$, $N_G(u_3) = N_G(u_4) = N_G(u_5) = 
	\{u_1,u_2\}$, 	then delete the edge $e = \{u_2,u_3\}$. 
\item[Operation 10.] If for two vertices $u$ and $v$ of $G$, $G - \{u,v\}$ has a connected component 
	$K$ with $|V(K)| \le 6$ such that $V(G) \ne V(K) \cup \{u,v\}$ and $G[V(K) \cup \{u,v\}]$ has 
	a Hamiltonian path $P$ from $u$ to $v$, then delete all edges of $G[V(K) \cup \{u,v\}]$ that 
	do not appear in $P$. 
\item[Operation 11.] If $G$ has an edge $e = \{u_1,u_2\}$ with $d_G(u_1) = d_G(u_2) = 2$, 
	then obtain $G_1$ from $G$ by merging $u_1$ and $u_2$ into a single vertex $u_1u_2$. 
\end{description} 

\begin{center}
\begin{figure}[htpb]
\centerline{\includegraphics[scale=.9]{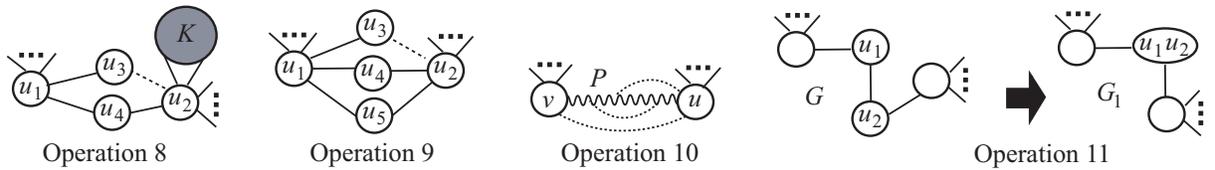}}
\caption{Operations 8 through~11, where the wavy line is a path and 
each dotted edge will be deleted.}
\label{fig:op8-11}
\end{figure}
\end{center}

\begin{lemma}\label{lem:op3}
Operation 8 is strongly safe.
\end{lemma}
\begin{proof}
If $e \not\in E(T)$, we are done. So, assume that $e \in E(T)$. Obviously, at least one vertex 
$v$ of $K$ is adjacent to $u_2$ in $T$ because $T$ is connected. So, $\{u_3,v\} \subseteq 
N_T(u_2)$. For each $i\in\{2,3\}$, let $T_i$ be the component tree of $T - \{e\}$ in which 
$u_i$ appears. If $u_4 \in V(T_3)$, then $u_4$ is a leaf of $T$ and hence adding the edge 
$\{u_2,u_4\}$ to $T - \{e\}$ clearly yields a spanning tree $\tilde{T}$ of $G$ with $|L(\tilde{T})| = 
|L(T)|$. So, we assume $u_4 \in V(T_2)$. If $u_1 \in V(T_3)$, then $u_4$ is a leaf of $T$ 
and hence adding the edge $\{u_1,u_4\}$ to $T - \{e\}$ clearly yields a spanning tree $\tilde{T}$ 
of $G$ with $|L(\tilde{T})| = |L(T)|$. Otherwise, $u_3$ is a leaf of $T$ and hence adding the edge 
$\{u_1,u_3\}$ to $T - \{e\}$ clearly yields a spanning tree $\tilde{T}$ of $G$ with $|L(\tilde{T})| = |L(T)|$. 
\end{proof}

\begin{lemma}\label{lem:op4}
Operation 9 is strongly safe.
\end{lemma}
\begin{proof}
If $e \not\in E(T)$, we are done. So, assume that $e \in E(T)$. Obviously, $\{u_4,u_5\} 
\cap L(T) \ne \emptyset$. Moreover, if for some $i\in\{4,5\}$, $u_i \in L(T)$ and 
$\{u_2,u_i\} \in E(T)$, then the proof of Lemma~\ref{lem:op3} shows that $T$ can be 
transformed into a spanning tree $\tilde{T}$ such that $|L(\tilde{T})| \le |L(T)|$ and $e \not\in E(\tilde{T})$. 
Thus, we may assume that $u_5 \in L(T)$, $\{u_1,u_4\} \in E(T)$, and $\{u_1,u_5\} \in E(T)$. 
Obviously, either $N_T(u_3) = \{u_2\}$ or $N_T(u_3) = \{u_1,u_2\}$. In the latter case, 
adding the edge $\{u_2,u_5\}$ to $T - \{e\}$ clearly yields a spanning tree $\tilde{T}$ of $G$, 
and $|L(\tilde{T})| = |L(T)|$ holds for $u_5 \in L(T)$. So, we assume the former case. 
Let $e' = \{u_1,u_5\}$. Then, adding the edges $\{u_1,u_3\}$ and $\{u_2,u_5\}$ 
to $T - \{e,e'\}$ clearly yields a spanning tree $\tilde{T}$ of $G$ with $|L(\tilde{T})| = |L(T)|$. 
\end{proof}

\begin{lemma}\label{lem:op5}
Operation 10 is strongly safe.
\end{lemma}
\begin{proof}
Operation 10 is clearly strongly safe if $V(G) = V(K) \cup \{u,v\}$. So, we assume that $V(G) \ne 
V(K) \cup \{u,v\}$. Since $K$ is a connected component, the degree of each vertex $x\not\in 
V(K)$ in $T - V(K)$ is $d_T(x)$ unless $x \in \{u,v\}$. 
Let $u \sim_T v$ be the path between $u$ and $v$ in $T$. 

Let $S$ be the set of internal vertices of $u \sim_T v$. Since $K$ is a connected component 
of $G - \{u,v\}$, either $S \cap V(K) = \emptyset$ or $S \subseteq V(K)$. Obviously, we are 
done if $S = V(K)$. So, we assume that $S \cap V(K)$ is either empty or contains at least 
one but not all vertices of $K$. Then, $T - \{u,v\}$ has one or more component trees in 
which at least one vertex of $K$ appears. Let $T_1$, \ldots, $T_\ell$ be such component 
trees. For each $i\in\{1,\ldots,\ell\}$, $V(T_i) \subseteq V(K)$ because $K$ is a connected 
component of $G - \{u,v\}$. Moreover, if $V(T_i) \ne S$, then $L(T) \cap V(T_i) \ne \emptyset$. 
Since $V(T_i) \ne S$ for at least one $i\in\{1,\ldots,\ell\}$, $|L(T) \cap V(K)| \ge 1$. 
Furthermore, if $S \cap V(K) = \emptyset$, then $|L(T) \cap V(K)| \ge \ell$.

{\em Case 1:} $S$ is a nonempty proper subset of $V(K)$. Then, modifying $T - V(K)$ by adding 
the edges of $P$ yields a new spanning tree $\tilde{T}$ of $G$. Clearly, $L(\tilde{T}) \setminus 
L(T) \subseteq \{u,v\}$. Moreover, since $V(G) \ne V(K) \cup \{u,v\}$, it is impossible that 
$\{u,v\} \subseteq L(\tilde{T})$. So, $|L(\tilde{T}) \setminus L(T)| \le 1$. 
Consequently, $|L(\tilde{T})| \le |L(T)|$ because $|L(T) \cap V(K)| \ge 1$. 

{\em Case 2:} $S \cap V(K) = \emptyset$. Then, both $u$ and $v$ are of degree at least~1 
in $T - V(K)$. We assume that the degree of $u$ in $T - V(K)$ is at least as large as 
that of $v$ in $T - V(K)$; the other case is similar. Let $y$ be the neighbor of $u$ in 
$u \sim_T v$. It is possible that $y = v$. Obviously, modifying $T - V(K)$ by adding the 
edges of $P$ and deleting the edge $\{u,y\}$ yields a new spanning tree $\tilde{T}$ of $G$. 
Clearly, $L(\tilde{T}) \setminus L(T) \subseteq \{u,y\}$. Thus, if $L(\tilde{T}) \setminus L(T) \ne 
\{u,y\}$, then $|L(\tilde{T})| \le |L(T)|$ because $|L(T) \cap V(K)| \ge 1$. Moreover, if $\ell 
\ge 2$, then $|L(T) \cap V(K)| \ge 2$ and in turn $|L(\tilde{T})| \le |L(T)|$. So, we may assume 
that $L(\tilde{T}) \setminus L(T) = \{u,y\}$ and $\ell = 1$. Then, the degree of $u$ in $T - V(K)$ 
is~1 and in turn so is $v$. Now, since $\ell=1$ and $u \in L(\tilde{T}) \setminus L(T)$, $v$ is 
adjacent to no vertex of $K$ in $T$ and hence $v$ is a leaf of $T$. Therefore,  
no matter whether $y = v$ or not, $|L(\tilde{T})| \le |L(T)|$ because $|L(T) \cap V(K)| \ge 1$. 
\end{proof}

\begin{lemma}\label{lem:op6}
Operation 11 is weakly safe.
\end{lemma}
\begin{proof}
For each $i\in\{1,2\}$, let $u'_i$ be the vertex in $N_G(u_i) \setminus \{u_{3-i}\}$. 
Possibly, $u'_1=u'_2$. If $u'_1\ne u'_2$, then $N_{G_1}(u_1u_2) = \{u'_1,u'_2\}$; 
otherwise, $N_{G_1}(u_1u_2) = \{u'_1\}$. 

First, we want to show that $opt(G) \le opt(G_1) + 1$. 
If $e \not\in E(T)$, then $T$ contains both $\{u'_1,u_1\}$ and $\{u'_2,u_2\}$ and 
we can modify $T$ (without decreasing $|L(T)|$) by replacing the edge $\{u'_2,u_2\}$ with $e$. 
So, we can assume that $e \in E(T)$. Then, it is clear that modifying $T$ by merging $u_1$ and 
$u_2$ into a single vertex $u_1u_2$ yields a spanning tree of $G_1$ whose weight is $w(T) - 1$. 
Thus, $opt(G) \le opt(G_1) + 1$. 

Next, suppose that $T_1$ is a spanning tree of $G_1$. If $u'_1=u'_2$, then $u_1u_2$ is 
a leaf of $T_1$ and its neighbor in $T_1$ is $u'_1$, and hence modifying $T_1$ by deleting 
the vertex $u_1u_2$ and adding the two edges $\{u_1,u_2\}$, $\{u_2,u'_2\}$ yields a spanning 
tree of $G$ whose weight is $w(T_1) + 1$. So, we assume that $u'_1 \ne u'_2$. Clearly, at 
least one of $\{u'_1,u_1u_2\}$ and $\{u'_2,u_1u_2\}$ is an edge of $T_1$. If for exactly one 
$i\in\{1,2\}$, $\{u'_i,u_1u_2\} \in E(T_1)$, then modifying $T_1$ by deleting the vertex 
$u_1u_2$ and adding the two edges $e$, $\{u'_i,u_i\}$ yields a spanning tree of $G$ 
whose weight is $w(T_1) + 1$. Otherwise, modifying $T_1$ by deleting the vertex $u_1u_2$ 
and adding the three edges $\{u'_1,u_1\}$, $e$, $\{u_2,u'_2\}$ yields a spanning tree 
of $G$ whose weight is $w(T_1) + 1$.
\end{proof}

Accordingly, we need to modify Step~\ref{step:ssafe} of the algorithm by choosing 
$i$ from $\{1,2,8,9,10\}$ and also modify Step~\ref{step:wsafe} by choosing 
$i$ from $\{3,4,11\}$. Obviously, after the modification, 
Steps~\ref{step:ssafe} and~\ref{step:wsafe} can be done in $O(n^2m)$ time.

\section{Computing a Preferred TFPCC ${\cal C}$}\label{sec:tfpcc}
In this section, we consider how to refine Step~\ref{step:mtfpc}. 
Because of Steps~\ref{step:ssafe} and~\ref{step:small}, we hereafter assume that 
$|V(G)| \ge 9$ and there is no $i\in\{1,\ldots,4,8,\ldots,11\}$ such that Operation~$i$ 
can be performed on $G$. Then, we can prove the next lemma:

\begin{lemma}\label{lem:cycle}
Suppose that $C$ is a cycle of $G$ with $|C| \le 8$. Let $A$ be the set of ports of $C$. 
%$u \in V(C)$ such that $N_G(u) \setminus V(C) \ne \emptyset$. 
Then, the following statements hold.
\begin{enumerate}
\item\label{stat:2+} $|A| \ge 2$. 
\item\label{stat:=2}  
	If $|A|=2$, then the two vertices in $A$ are not adjacent in $C$ and $|C| \ne 5$. 
\item\label{stat:2and4} 
	If $|A|=2$ and $|C|=4$, then $G[V(C)]$ and $C$ are the same graph. 
\end{enumerate}
\end{lemma}
\begin{proof}
We prove the statements separately as follows.

{\em Statement~\ref{stat:2+}:} Since $G$ is connected and $|C| < 9 \le |V(G)|$, 
$|A| \ge 1$. Moreover, since Operation~4 cannot be performed on $G$, $|A| \ge 2$. 

{\em Statement~\ref{stat:=2}:} Suppose that $|A| = 2$. Then, the two vertices in $A$ 
cannot be adjacent in $C$, because otherwise Operation~10 %or~6 
could be performed on $G$. For a contradiction, assume that $|C| = 5$. Suppose that 
$u_1$, \ldots, $u_5$ are the vertices of a 5-cycle of $C$ and appear in $C$ clockwise 
in this order. Since the two vertices in $A$ are not adjacent in $C$, we may assume that 
$A = \{u_1,u_3\}$. If $\{u_2,u_4\} \in E(G)$ or $\{u_2,u_5\} \in E(G)$, then 
Operation~10 can be performed on $G$, a contradiction. 
So, we assume that $\{u_2,u_4\} \not\in E(G)$ and $\{u_2,u_5\} \not\in E(G)$. 
If $\{u_1,u_4\} \in E(G)$ or $\{u_3,u_5\} \in E(G)$, then Operation~10 can be performed 
on $G$, a contradiction. Thus, we may further assume that $\{u_1,u_4\} \not\in E(G)$ 
and $\{u_3,u_5\} \not\in E(G)$. Now, $\{u_4,u_5\} \in E(G)$, $d_G(u_4) = 2$, and 
$d_G(u_5) = 2$. Hence, Operation~11 can be performed on $G$, a contradiction. 

{\em Statement~\ref{stat:2and4}:} Suppose that $|A|=2$ and $|C| = 4$. The two vertices 
in $A$ are not adjacent in $C$ by Statement~\ref{stat:=2}, and hence $G[V(C)]$ and $C$ 
are the same graph because otherwise Operation~10 could be performed on $G$. 
\end{proof}

To refine Step~\ref{step:mtfpc}, our idea is to compute ${\cal C}$ as a {\em preferred} 
TFPCC of $G$. Before defining what the word ``preferred'' means here, we need to prove a lemma. 
For ease of explanation, we assume, with loss of generality, that there is a linear order 
(denoted by $\prec$) on the vertices of $G$. 

\begin{lemma}\label{lem:exclude}
Suppose that $u_1$ and $u_3$ are two vertices of $G$ such that $u_1\prec u_3$ and 
Condition~C1 below holds. Then, $G$ has an OST in which $u_1$ or $u_3$ is a leaf. 
Consequently, $G$ has an OST in which $u_1$ is a leaf. 
\begin{itemize}
\item[C1.] For two vertices $u_2$ and $u_4$ in $V(G) \setminus \{u_1,u_3\}$, 
	$N_G(u_1) = N_G(u_3) = \{u_2,u_4\}$. 
\end{itemize}
\end{lemma}
\begin{proof}
If $u_1$ is a leaf of $T$, then we are done. 
So, assume that $u_1$ is not a leaf of $T$. Since Condition~C1 holds, 
$u_3$ is clearly a leaf of $T$ and we can modify $T$ (without decreasing $w(T)$) 
by switching $u_1$ and $u_3$ so that $u_1$ becomes a leaf in $T$. 
\end{proof}

If Condition~C1 in Lemma~\ref{lem:exclude} holds for $u_1$ and $u_3$, we refer to 
$u_2$ and $u_4$ as the {\em boundary points} of the pair $p=(u_1,u_3)$, and refer to 
the edges incident to $u_1$ or $u_3$ as the {\em supports} of $p$. 

Let $\Pi$ be the set of pairs $(u_1,u_3)$ of vertices in $G$ satisfying Condition~C1. 
It is worth pointing out that for each $p \in \Pi$ and each boundary point $u$ of $p$, 
$d_G(u) \ge 3$ because otherwise Operation~4 could be performed on $G$. 

\begin{lemma}\label{lem:ind}
No two pairs in $\Pi$ share a support. 
\end{lemma}
\begin{proof}
Obviously, for two pairs in $\Pi$ to share a support, they have to share their boundary 
points. However, no two pairs in $\Pi$ can share their boundary points, because otherwise 
Operation~9 could be performed on $G$. So, no two pairs in $\Pi$ share a support. 
\end{proof}

\begin{lemma}\label{lem:excludeAll}
$G$ has an OST in which $u_1$ is a leaf for each $(u_1,u_3) \in \Pi$. 
\end{lemma}
\begin{proof}
By Lemma~\ref{lem:exclude}, we can assume that for every $p = (u_1,u_3) \in \Pi$, 
$d_T(u_1) \le 1$. In a nutshell, the proof of Lemma~\ref{lem:exclude} shows that 
even if $T$ is an OST with $d_T(u_1) \ge 2$ , we can modify $T$ without decreasing 
$w(T)$ so that $d_T(u_1) \le 1$. Indeed, the modification only uses the supports of $p$. 
Now, by Lemma~\ref{lem:ind}, a similar modification can be done independently 
for each other $p' \in \Pi$. Therefore, the lemma holds. 
\end{proof}

Now, we are ready to make two definitions. Let ${\cal C}$ be a TFPCC of $G$. 
${\cal C}$ is {\em special} if for every pair $(u_1,u_3) \in \Pi$, 
$d_{\cal C}(u_1) \le 1$. ${\cal C}$ is {\em preferred} if ${\cal C}$ is special 
and $|E({\cal C})|$ is maximized over all special TFPCCs of $G$. 

\begin{lemma}\label{lem:up}
If ${\cal C}$ is a preferred TFPCC of $G$, then $opt(G) \le |E({\cal C})|$. 
\end{lemma}
\begin{proof}
By Lemma~\ref{lem:excludeAll}, $G$ has an OST $\tilde{T}$ such that for each $(u_1,u_3) 
\in \Pi$, $d_{\tilde{T}}(u_1) = 1$. So, by Lemma~\ref{lem:LZW}, we can construct a path 
cover ${\cal P}$ of $G$ with $|E({\cal P})| \ge w(\tilde{T})$ such that $d_{\cal P}(u_1) 
\le 1$ for every $(u_1,u_3) \in \Pi$. Thus, ${\cal P}$ is a special TFPCC of $G$. 
Consequently, if ${\cal C}$ is a preferred TFPCC of $G$, then $opt(G) = w(\tilde{T}) 
\le |E({\cal P})| \le |E({\cal C})|$.
\end{proof}

\begin{lemma}\label{lem:prefer}
We can compute a preferred TFPCC ${\cal C}$ of $G$ in $t(2n, 2m)$ time. 
\end{lemma}
\begin{proof}
We construct a new graph $G'$ from $G$ by adding a new vertex $x_p$ and the edge $\{u_1,x_p\}$ 
for each pair $p = (u_1,u_3) \in \Pi$. Obviously, if ${\cal C}^*$ is a preferred TFPCC of $G$, 
then adding the edges $\{x_p,u_1\}$ with $p=(u_1,u_3) \in \Pi$ to ${\cal C}^*$ yields a TFPCC 
${\cal C}'$ of $G'$ with $|E({\cal C}')| = |E({\cal C}^*)| + |\Pi|$. 

We then compute a maximum TFPCC ${\cal C}'$ of $G'$ in $t(|V(G')|,|E(G')|)$ time. 
By the discussion in the last paragraph, $|E({\cal C}')| \ge |E({\cal C}^*)| + |\Pi|$. 
If for some $p = (u_1,u_3) \in \Pi$, $d_{{\cal C}'}(x_p) = 0$, then by the maximality of ${\cal 
C}'$, $d_{{\cal C}'}(u_1) = 2$ and we can modify ${\cal C}'$ by replacing one of the edges 
incident to $u_1$ in ${\cal C}'$ with the edge $\{x_p,u_1\}$. Clearly, ${\cal C}'$ is still 
a maximum TFPCC of $G'$ after the modification. So, we can repeatedly modify ${\cal C}'$ in 
this way until $d_{{\cal C}'}(x_p) = 1$ for every $p=(u_1,u_3) \in \Pi$. ${\cal C}'$ is now 
a maximum TFPCC of $G'$ such that for every $p=(u_1,u_3) \in \Pi$, $d_{{\cal C}'}(x_p) = 1$. 
Finally, we obtain ${\cal C}$ from ${\cal C}'$ by deleting the edge $\{x_p,u_1\}$ for each 
$p = (u_1,u_3) \in \Pi$. Clearly, $|E({\cal C})| = |E({\cal C}')| - |\Pi| \ge |E({\cal C}^*)|$. 
Therefore, ${\cal C}$ is a preferred TFPCC of $G$. Since $|V(G')| \le 2|V(G)|$ and 
$|E(G')| \le 2|E(G)|$, the lemma holds.
\end{proof}

Recall that $t(n,m)=O(n^2m^2)$ \cite{Har84}. So, Lemma~\ref{lem:prefer} ensures that 
after modifying Step~\ref{step:mtfpc} by computing ${\cal C}$ as a preferred TFPCC of $G$, 
Step~\ref{step:mtfpc} can still be done in $t(n,m)$ time.

\section{Preprocessing ${\cal C}$}\label{sec:prep}
In this section, we consider how to refine Step~\ref{step:prep}. So, suppose that we have 
computed a preferred TFPCC ${\cal C}$ of $G$ as in Lemma~\ref{lem:prefer}. To refine 
Step~\ref{step:prep}, we repeatedly perform not only Operations~5 through~7 but also 
the following three operations on ${\cal C}$ until none of the six is applicable. 

\begin{description}
\item[Operation 12.]
	If a cycle $C_1$ of ${\cal C}$ has an edge $e_1=\{u_1,u'_1\}$ and another cycle or 
	path component $C_2$ of ${\cal C}$ has an edge $e_2=\{u_2,u'_2\}$ such that 
	$e=\{u_1,u_2\} \in E(G)$ and $e'=\{u'_1,u'_2\} \in E(G)$, then combine $C_1$ and $C_2$ 
	into a single cycle or path by replacing $e_1$ and $e_2$ with $e$ and $e'$. 
\item[Operation 13.]
	If an endpoint $u_1$ of a path component $P_1$ of ${\cal C}$ is adjacent to an endpoint 
	$u_2$ of another path component $P_2$ of ${\cal C}$ in $G$, then combine $P_1$ and $P_2$ 
	into a single path by adding the edge $\{u_1,u_2\}$. 
\item[Operation 14.]
	If $e=\{u,v\}$ is an edge of a path component of ${\cal C}$ such that for some isolated 
	vertex $x$ of ${\cal C}$, $\{u,x\}\in E(G)$ and $\{v,x\}\in E(G)$, then replace $e$ by 
	the edges $\{u,x\}$ and $\{v,x\}$.
\end{description}

\begin{center}
\begin{figure}[htpb]
\centerline{\includegraphics[scale=.9]{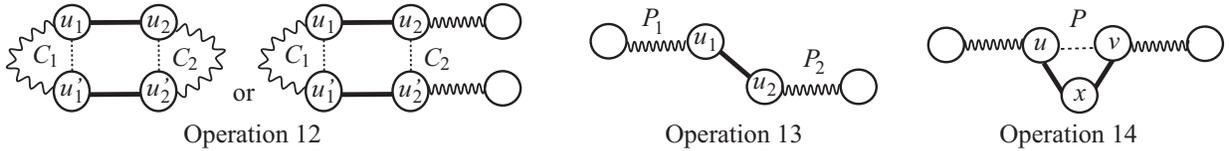}}
\caption{Operations 12 through 14, where each wavy line or curve is a path, 
each dotted edge will be deleted, and each bold edge will be added.}
\label{fig:op12-14}
\end{figure}
\end{center}

\begin{lemma}\label{lem:prep}
Immediately after the refined preprocessing step, the following statements hold:
\begin{enumerate}
\item\label{stat:p1}
	${\cal C}$ is a TFPCC of $G$ and has at least $opt(G)$ edges. 
\item\label{stat:p3} 
	If a path component $P$ of ${\cal C}$ is of length at most~3, then $P$ is alive. 
\item\label{stat:p5} 
	If an endpoint $v$ of a path component $P$ of ${\cal C}$ is a port of $P$, then 
	each vertex in $N_G(v) \setminus V(P)$ is an internal vertex of a path component $Q$ 
	of ${\cal C}$ with $|Q| \ge 2|P| + 2$.
\item\label{stat:p2} 
	No pair $(u_1,u_3) \in \Pi$ satisfies that $u_1$ appears in a cycle of ${\cal C}$. 
\item\label{stat:p4} 
	If a dead path component $P$ of ${\cal C}$ is of length~4, 
	then both endpoints of $P$ are leaves in $G$. 
\item\label{stat:p6}
	Each 4-cycle $C$ of ${\cal C}$ has at least three ports. 
\end{enumerate}
\end{lemma}
\begin{proof}
A {\em short} cycle is a cycle of length at most~7. 
We prove the statements separately as follows.

{\em Statement~\ref{stat:p1}:} Before the refined preprocessing, ${\cal C}$ has at least $opt(G)$ 
edges by Lemma~\ref{lem:up} and is a TFPCC of $G$. Since Operation~$i$ does not 
decrease the number of edges in ${\cal C}$ or creates a new short cycle or a vertex of 
degree larger than~2 in ${\cal C}$ for each $i \in \{5,6,7,12,13,14\}$, 
Statement~\ref{stat:p1} holds. 

{\em Statement~\ref{stat:p3}:} 
Same as that of Statement~\ref{stat:q3} in Lemma~\ref{lem:prep0}. 

{\em Statement~\ref{stat:p5}:} 
	Suppose that an endpoint $v$ of a path component $P$ of ${\cal C}$ is a port. 
	Consider an arbitrary $u \in N_G(v) \setminus V(P)$. 
	Since neither Operation~6 nor Operation~13 can be performed on ${\cal C}$, $u$ is 
	an internal vertex of a path component $Q$ of ${\cal C}$. Let $u_1$ and $u_2$ be the 
	endpoints of $Q$. For each $i\in\{1,2\}$, let $Q_i$ be the path from $u$ to $u_i$ in $P$. 
	Then, $|Q| = |Q_1| + |Q_2|$. Moreover, since Operation~7 cannot be applied on ${\cal C}$, 
	$|P| + |Q_i| + 1 \le |Q|$ for each $i\in\{1,2\}$. Thus, $2|P| + 2 \le |Q|$.

{\em Statement~\ref{stat:p2}:} Before the refined preprocessing, no pair $(u_1,u_3) \in \Pi$ 
satisfies that $u_1$ appears in a cycle of ${\cal C}$ because ${\cal C}$ is a preferred 
TFPCC of $G$. Moreover, if Operation~$i$ creates a new cycle $C$ in ${\cal C}$ for some 
$i\in\{5,6,7,12,13,14\}$, then $i = 12$ and $C$ is obtained by merging two shorter cycles 
in ${\cal C}$. Thus, Statement~\ref{stat:p2} holds. 

{\em Statement~\ref{stat:p4}:} Let $P$ be a dead path component of ${\cal C}$ with 
$|P| = 4$. Suppose that $u_1$, \ldots, $u_5$ are the vertices of $P$ and they appear 
in $P$ in this order. If all internal vertices of $P$ are ports, then both $u_1$ and 
$u_5$ are leaves of $G$ (and we are done), because otherwise Operation~5 could be 
performed on ${\cal C}$. Moreover, if at most one internal vertex of $P$ is a port, 
then $G$ would be disconnected or Operation~4 could be performed on $G$, 
a contradiction. So, we assume that exactly two internal vertices of $P$ are ports. 
Now, if $\{i,j\} = \{2,4\}$, then $\{u_1,u_5\} \not\in E(G)$, $\{u_1,u_3\}\not\in E(G)$, 
and $\{u_5,u_3\}\not\in E(G)$ (because otherwise Operation~5 could be performed on 
${\cal C}$), and in turn both $u_1$ and $u_5$ are leaves of $G$ (and we are done) 
because otherwise Operation~8 or~9 could be performed on $G$. Thus, 
we may assume that $i = 2$ and $j = 3$. Then, since Operation~5 cannot be performed 
on ${\cal C}$, $u_1$ is a leaf of $G$ and $\{u_1,u_2\} \cap N_G(u_5) = \emptyset$. 
For the same reason, $\{u_3,u_5\} \not\in E(G)$ or $\{u_2,u_4\} \not\in E(G)$. 
Indeed, $\{u_2,u_4\} \in E(G)$ because otherwise the edge $\{u_2,u_3\}$ would be 
deleted by Operation~2 or~3, Therefore, $\{u_3,u_5\} \not\in E(G)$ and in turn 
$u_5$ is also a leaf of $G$. 

{\em Statement~\ref{stat:p6}:} 
Let $C$ be a 4-cycle in ${\cal C}$, and $A$ be the set of ports of $C$. 
Further let $u_1$, \ldots, $u_4$ be the vertices of $C$ and assume that 
they appear in $C$ clockwise in this order. By Lemma~\ref{lem:cycle}, $|A|\ge 2$. 
For a contradiction, assume that $|A| = 2$. Then, by Statement~\ref{stat:2and4} in 
Lemma~\ref{lem:cycle}, $A = \{u_1,u_3\}$ or $A = \{u_2,u_4\}$. We may assume that 
$A = \{u_2,u_4\}$ and $u_1 \prec u_3$. Then, $N_G(u_1) = N_G(u_3) = \{u_2,u_4\}$ by 
Statement~3 in Lemma~\ref{lem:cycle}, and in turn $(u_1,u_3) \in \Pi$. 
Since the refined preprocessing of ${\cal C}$ does not introduce a new short cycle, 
$C$ is a cycle in ${\cal C}$ even before the refined preprocessing. However, this contradicts 
the fact that ${\cal C}$ is a preferred TFPCC of $G$ before the refined preprocessing. 
\end{proof}

Obviously, the refined preprocessing (i.e., Step~\ref{step:prep}) can be done in $O(nm)$ time.

\section{Transforming ${\cal C}$ into a Spanning Tree}\label{sec:trans}
In this section, we consider how to refine Step~\ref{step:trans}. So, 
suppose that we have just performed the refined preprocessing on ${\cal C}$ as in 
Section~\ref{sec:prep}. Let $\Gamma$ be the set of (ordered) pairs $(P,Q)$ of 
path components of ${\cal C}$ such that $|P| \ge 1$ and some endpoint $v$ of $P$ 
is adjacent to a vertex $u$ of $Q$ in $G$. Note that $d_{\cal C}(u) = 2$ and 
$2|P| + 2 \le |Q|$ by Statement~\ref{stat:p5} in Lemma~\ref{lem:prep}. 
Suppose that we obtain a subset $\Gamma'$ of $\Gamma$ from $\Gamma$ as follows. 
\begin{itemize}
\item For each path component $P$ of ${\cal C}$ such that there are two or more path 
	components $Q$ of ${\cal C}$ 	with $(P,Q) \in \Gamma$, delete all but one pair $(P,Q)$ 
	from $\Gamma$. 
\end{itemize}
Now, consider an auxiliary digraph $D$ such that the vertices of $D$ one-to-one correspond 
to the path components $P$ of ${\cal C}$ with $|P|\ge 1$ and the arcs of $D$ one-to-one 
correspond to the pairs in $\Gamma'$. By Statement~\ref{stat:p5} in Lemma~\ref{lem:prep}, 
$D$ is a rooted forest (in which each leaf is of in-degree~0, each root is of out-degree~0, 
and each vertex is of out-degree at most~1). 

To transform ${\cal C}$ into a spanning tree of $G$, the idea is to modify ${\cal C}$ in three 
stages. ${\cal C}$ is initially a TFPCC of $G$ and we will always keep ${\cal C}$ being a TFTCC 
of $G$. For each $i\in\{1,2,3\}$, we use ${\cal C}_i$ to denote the ${\cal C}$ immediately 
after the $i$-th stage. For convenience, we use ${\cal C}_0$ to denote the ${\cal C}$ 
immediately before the first stage. Moreover, for each $i \in \{1,2,3\}$ and each connected 
component $C$ of ${\cal C}_i$, we use $b(C)$ to denote the number of edges $\{u,v\} \in 
E({\cal C}_0)$ such that $\{u,v\}\subseteq V(C)$. 

In the first stage, we modify ${\cal C}$ by performing the following step:

\begin{enumerate}
\item\label{step:conPath}
	For each pair $(P,Q) \in \Gamma'$, add an arbitrary $\{u,v\} \in E(G)$ to ${\cal C}$ 
	such that $u$ is an endpoint of $P$ and $v$ appears in $Q$. 
\end{enumerate}

\begin{lemma}\label{lem:type1}
Each connected component of ${\cal C}_1$ that is not a path or cycle is 
a tree $\hat{T}$ satisfying Condition C2 below:
\begin{itemize}
\item[C2.] $b(\hat{T}) \ge 5$, $|L(\hat{T})| \le b(\hat{T}) - 2$, and 
$w(\hat{T}) \ge \frac{4}{5} b(\hat{T})$. 
\end{itemize} 
\end{lemma}
\begin{proof}
Let $\hat{T}$ be a connected component of ${\cal C}_1$ that is not a path or cycle. 
Obviously, $\hat{T}$ can be obtained from a tree component $\hat{T}_D$ of $D$ by replacing each vertex 
of $\hat{T}_D$ with the corresponding path component of ${\cal C}$ and replacing each arc of $\hat{T}_D$ 
corresponding to a pair $(P,Q) \in \Gamma'$ with an edge $\{v,u\} \in E(G)$ such that 
$v$ is an endpoint of $P$ and $u$ appears in $Q$. Thus, $\hat{T}$ is clearly a tree. 

We next prove that $\hat{T}$ satisfies Condition~C2 by induction on the number of arcs in $\hat{T}_D$. 
Clearly, $\hat{T}_D$ has at least one edge. In the base case, $\hat{T}_D$ has only one arc. Let $(P,Q)$ 
be the pair in $\Gamma'$ corresponding to the arc. $\hat{T}$ is obtained from $P$ and $Q$ by 
connecting them with an edge $\{v,u\}\in E(G)$ such that $v$ is an endpoint of $P$ and $u$ 
appears in $Q$. Thus, $w(\hat{T}) = |P| + |Q| - 1$, $|L(\hat{T})| = 3$, and $b(\hat{T}) = |P| + |Q|$. Hence, 
by Statement~\ref{stat:p5} in Lemma~\ref{lem:prep}, $b(\hat{T}) \ge 3|P| + 2\ge 5$. Therefore, 
$|L(\hat{T})| \le b(\hat{T}) - 2$ and $\frac{w(\hat{T})}{b(\hat{T})} = 1 - \frac{1}{b(\hat{T})} \ge \frac{4}{5}$. 
This shows that $\hat{T}$ satisfies Condition~C2 in the base case. 

Now, assume that $\hat{T}_D$ has at least two arcs. Consider an arbitrary $(P,Q) \in \Gamma'$ 
such that the vertex $\alpha$ of $D$ corresponding to $P$ is a leaf of $\hat{T}_D$. Let $\hat{T}'_D$ 
be obtained from $\hat{T}_D$ by deleting $\alpha$, and $\hat{T}'$ be obtained from $\hat{T}$ by deleting 
the vertices of $P$. Since $\hat{T}'_D$ has one fewer arc than $\hat{T}_D$, the inductive hypothesis 
implies that $b(\hat{T}') \ge 5$, $|L(\hat{T}')| \le b(\hat{T}') - 2$, and $w(\hat{T}') \ge \frac{4}{5} b(\hat{T}')$. 
Obviously, $b(\hat{T}) = b(\hat{T}') + |P|$, $|L(\hat{T})| = |L(\hat{T}')| + 1$, and $w(\hat{T}) = w(\hat{T}') + |P|$. 
Since $|P| \ge 1$, it is now easy to verify that $\hat{T}$ satisfies Condition~C2.
\end{proof}

Hereafter, a connected component of ${\cal C}$ is {\em good} if it is a tree $\hat{T}$ 
satisfying Condition~C2 in Lemma~\ref{lem:type1} or Condition~C3 below, 
while it is {\em bad} otherwise. 
\begin{itemize}
\item[C3.] $w(\hat{T}) \ge b(\hat{T}) = 4$ and $|L(\hat{T})| = 3$. 
\end{itemize} 

\begin{lemma}\label{lem:type2}
Suppose that $C$ is a bad connected component of ${\cal C}_1$. Then, $C$ is a cycle of 
length at least~4, a 0-path, or a 4-path whose endpoints are leaves of $G$. Moreover, 
if $C$ is a 0-path, then the unique vertex $u \in V(C)$ satisfies that each $v \in 
N_G(u)$ is an internal vertex of a tree component of ${\cal C}_1$ and no two vertices 
in $N_G(u)$ are adjacent in ${\cal C}_1$. 
\end{lemma}
\begin{proof}
Since $C$ is bad, Lemma~\ref{lem:type1} ensures that $C$ is a path or cycle and in turn 
is a connected component of ${\cal C}_0$. Indeed, $C$ cannot be a path of length at least~5, 
because otherwise $C$ would satisfy Condition~C2. Now, by Lemma~\ref{lem:prep}, $C$ is 
a cycle of length at least~4, a 0-path, or a 4-path whose endpoints are leaves of $G$. 

Suppose that $C$ is a 0-path. Then, $C$ is also 0-path in ${\cal C}_0$. Let $u$ 
be the unique vertex in $C$. Consider an arbitrary $v \in N_G(u)$. Since Operation~13 
cannot be performed on ${\cal C}_0$, $v$ is not a leaf of a tree component of 
${\cal C}_1$. Moreover, since Operation~6 cannot be performed on ${\cal C}_0$, 
$v$ does not appear in a cycle of ${\cal C}_1$. Furthermore, since Operation~14 cannot 
be performed on ${\cal C}_0$, no two vertices in $N_G(u)$ are adjacent in ${\cal C}_1$. 
\end{proof}

We next want to define several operations on ${\cal C}$ none of which will produce 
a new cycle or a new bad connected component in ${\cal C}$. An operation on ${\cal C}$ 
is {\em good} if it either just connects two or more connected components of ${\cal C}$ 
into a single good connected component, or modify a good connected component of 
${\cal C}$ so that it has more internal vertices (and hence remains good). 

In the second stage, we modify ${\cal C}$ by repeatedly performing the following 
operations on ${\cal C}$ until none of them is applicable. 

\begin{description}
\item[Operation 15.] 
	If ${\cal C}$ has two cycles $C_1$ and $C_2$ such that $|C_1| + |C_2| \ge 10$ and 
	some edge $e=\{v_1,v_2\}$ of $G$ satisfies $v_1\in V(C_1)$ and $v_2 \in V(C_2)$, then 
	connect $C_1$ and $C_2$ into a single path $T$ by deleting one edge incident to $v_1$ 
	in $C_1$, deleting one edge incident to $v_2$ in $C_2$, and adding the edge $e$. 
\item[Operation 16.] 
	If ${\cal C}$ has a cycle $C_1$ of length at least~5 and a good connected component 
	$C_2$ such that some edge $e=\{v,u\}$ of $G$ satisfies $v\in V(C_1)$ and $u \in V(C_2)$,
	then connect $C_1$ and $C_2$ into a single tree $T$ by deleting one edge incident to $v$ 
	in $C_1$ and adding the edge $e$. 
\item[Operation 17.] 
	If ${\cal C}$ has a cycle $C$ of length at least~6 and a 4-path component $P$ 
	such that some edge $e=\{v,u\}$ of $G$ satisfies $v\in V(C)$ and $u \in V(P)$, then 
	connect $C$ and $P$ into a single tree $T$ by deleting one edge incident to $v$ in 
	$C$ and adding the edge $e$. 
\item[Operation 18.] 
	If ${\cal C}$ has a 0-path component $P$ whose unique vertex $u$ has two neighbors 
	$v_1$ and $v_2$ in $G$ such that $v_1$ and $v_2$ fall into different connected components
	$C_1$ and $C_2$ of ${\cal C}$, then connect $P$, $C_1$, and $C_2$ into a single 
	connected component $T$ by adding the edges $\{u,v_1\}$ and $\{u,v_2\}$. 
\item[Operation 19.] 
	If ${\cal C}$ has a good connected component $C_1$ and another connected component 
	$C_2$ such that some leaf $u$ of $C_1$ is adjacent to a vertex $v$ of $C_2$ in $G$, 
	then connect $C_1$ and $C_2$ into a single tree component $T$ by deleting one edge 
	incident to $v$ in $C_2$ if $C_2$ is a cycle, and further adding the edge $\{u,v\}$. 
\item[Operation 20.] 
	If a cycle $C$ of ${\cal C}$ has an edge $e=\{v_1,v_2\}$ such that some $u_1 \in 
	N_G(v_1) \setminus V(C)$ and some $u_2 \in N_G(v_2) \setminus V(C)$ fall into 
	different connected components $C_1$ and $C_2$ of ${\cal C}$ other than $C$, 
	then connect $C$, $C_1$, and $C_2$ into a single tree component $T$ by deleting $e$, 
	deleting one edge incident to $u_1$ if $C_1$ is a cycle, deleting one edge incident 
	to $u_2$ if $C_2$ is a cycle, and adding the edges $\{v_1,u_1\}$ and $\{v_2,u_2\}$.
\item[Operation 21.]
	If a good connected components $C$ of ${\cal C}$ is not a Hamiltonian path of $G$ 
	but is a dead path whose endpoints are adjacent in $G$, then choose an arbitrary 
	port $u$ of $C$, modify $C$ by adding the edge of $G$ between the endpoints of 
	$C$ and deleting one edge incident to $u$ in $C$, and further perform Operation~19.
\item[Operation 22.] 
	If a good connected component $C$ of ${\cal C}$ is not a path but has two leaves 
	$u$ and $v$ with $\{u,v\} \in E(G)$, then modify $C$ by first finding an arbitrary 
	vertex $x$ on the path $P$ between $u$ and $v$ in $C$ with $d_C(x) \ge 3$, then 
	deleting one edge incident to $x$ in $P$, and further adding the edge $\{u,v\}$. 
\item[Operation 23.] 
	If ${\cal C}$ has a 0-path component $C_1$, a 4-path component $P$, and a connected 
	component $C_2$ other than $C_1$ and $P$ such that the center vertex $u_3$ of $P$ 
	is adjacent to a vertex $x$ of $C_2$ in $G$ and the unique vertex $v$ of $C_1$ is 
	adjacent to the other two internal vertices $u_2$ and $u_4$ of $P$ (than $u_3$) in 
	$G$, then connect $C_1$, $P$, and $C_2$ into a single connected component $T$ by 
	deleting the edge $\{u_2,u_3\}$, deleting one edge incident to $x$ if $C_2$ is 
	a cycle, and adding the edges $\{v,u_2\}$, $\{v,u_4\}$, $\{u_3,x\}$. 
\end{description}

\begin{center}
\begin{figure}[htpb]
\centerline{\includegraphics[scale=.9]{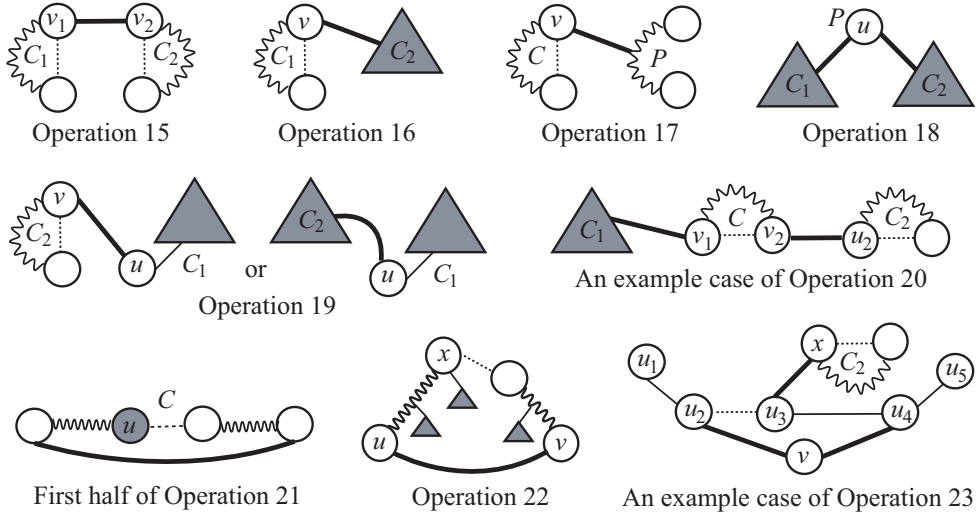}}
\caption{Operations 15 through 23, where the filled circle is a port, 
each wavy line or curve is a path, each filled triangle is a tree, 
each dotted edge will be deleted, and each bold edge will be added.}
\label{fig:op15-23}
\end{figure}
\end{center}

In the following proofs of Lemmas~\ref{lem:op15} through~\ref{lem:op23}, 
$\hat{T}$ denotes the new connected component of ${\cal C}$ created 
by the corresponding operation. 

\begin{lemma}\label{lem:op15}
Operation 15 is good.
\end{lemma}
\begin{proof}
Obviously, $w(\hat{T}) = |C_1| + |C_2| - 2$, $b(\hat{T}) = |C_1| + |C_2| \ge 10$, 
and $|L(\hat{T})| = 2$. Thus, $\hat{T}$ is good.  
\end{proof}

\begin{lemma}\label{lem:op16}
Operation 16 is good.
\end{lemma}
\begin{proof}
Obviously, $w(\hat{T}) \ge w(C_2) + |C_1| - 1 \ge \frac{4}{5}b(C_2) + |C_1| - 1$, 
$b(\hat{T}) = b(C_2) + |C_1| \ge b(C_2) + 5$, and $|L(\hat{T})| \le |L(C_2)| + 1$. 
Thus, $\hat{T}$ is good.  
\end{proof}

\begin{lemma}\label{lem:op17}
Operation 17 is good.
\end{lemma}
\begin{proof}
Obviously, $w(\hat{T}) = |C| + 2$, $b(\hat{T}) = |C| + 4 \ge 10$, and $|L(\hat{T})| = 3$. 
Thus, $\hat{T}$ is good.  
\end{proof}

\begin{lemma}\label{lem:op18}
Operation 18 is good.
\end{lemma}
\begin{proof}
Since Operation~6 cannot be applied on ${\cal C}_0$, neither $C_1$ nor $C_2$ is a cycle. 
Hence, both $C_1$ and $C_2$ are trees and in turn $\hat{T}$ is a tree. 
To show that $\hat{T}$ is good, we distinguish three cases as follows.

{\em Case 1:} Both $C_1$ and $C_2$ are good. In this case, $w(\hat{T}) \ge w(C_1) + w(C_2) + 1
\ge \frac{4}{5}b(C_1) + \frac{4}{5}b(C_2) + 1$, $b(\hat{T}) = b(C_1) + b(C_2) \ge 8$, and 
$|L(\hat{T})| \le |L(C_1)| + |L(C_2)|$. Thus, $\hat{T}$ is clearly good. 

{\em Case 2:} One of $C_1$ and $C_2$ is good. W.l.o.g., we assume that $C_1$ is good and 
$C_2$ is bad. Then, by Lemma~\ref{lem:type2}, $C_2$ is either a 0-path or a 4-path whose 
endpoints are leaves of $G$. The former case is impossible, because Operation~13 cannot 
be performed on ${\cal C}_0$. In the latter case, 
$w(\hat{T}) \ge w(C_1) + 4 \ge \frac{4}{5}b(C_1) + 4$, $b(\hat{T}) = b(C_1) + 4 \ge 8$, 
and $|L(\hat{T})| \le |L(C_1)| + 2$, implying that $\hat{T}$ is good. 

{\em Case 3:} Both $C_1$ and $C_2$ are bad. In this case, both $C_1$ and $C_2$ are 4-paths 
whose endpoints are leaves of $G$, because Operation~13 cannot be performed on ${\cal C}_0$. 
So, $w(\hat{T}) = 7$, $b(\hat{T}) = 8$, and $|L(\hat{T})| = 4$, implying that $\hat{T}$ is good. 
\end{proof}

\begin{lemma}\label{lem:op19}
Operation 19 is good.
\end{lemma}
\begin{proof}
$\hat{T}$ is clearly a tree. To show that $\hat{T}$ is good, we distinguish three cases as follows.

{\em Case 1:} $C_2$ is a cycle. In this case, $w(\hat{T}) = w(C_1) + |C_2| \ge \frac{4}{5}b(C_1) 
+ |C_2|$, $b(\hat{T}) = b(C_1) + |C_2| \ge 8$, and $|L(\hat{T})| = |L(C_1)|$. So, $\hat{T}$ is clearly good. 

{\em Case 2:} $C_2$ is good. In this case, $w(\hat{T}) \ge w(C_1) + w(C_2) + 1 \ge 
\frac{4}{5}b(C_1) + \frac{4}{5}b(C_2) + 1$, $b(\hat{T}) = b(C_1) + b(C_2) \ge 8$, and 
$|L(\hat{T})| = |L(C_1)| + |L(C_2)| - 1$. So, $\hat{T}$ is clearly good. 

{\em Case 3:} $C_2$ is bad but not a cycle. In this case, Lemma~\ref{lem:type2} ensures that 
$C_2$ is either a 0-path or a 4-path whose endpoints are leaves of $G$. In the latter case, 
$w(\hat{T}) \ge w(C_1) + 4 \ge \frac{4}{5}b(C_1) + 4$, $b(\hat{T}) = b(C_1) + 4 \ge 8$, and $|L(\hat{T})| = 
|L(C_1)| + 1$, implying that $\hat{T}$ is clearly good. So, we assume the former case. If $C_1$ 
satisfies Condition~C2, then $w(\hat{T}) = w(C_1) + 1 \ge \frac{4}{5}b(C_1) + 1$, $b(\hat{T}) = b(C_1)
\ge 5$, and $|L(\hat{T})| = |L(C_1)|$, implying that $\hat{T}$ is good. Otherwise, $b(\hat{T}) = b(C_1) = 4$, 
$w(\hat{T}) = w(C_1) + 1 \ge b(C_1) + 1 > b(\hat{T})$, and $|L(\hat{T})| = |L(C_1)|$, implying that $\hat{T}$ is good.
\end{proof}

\begin{lemma}\label{lem:op20}
Operation 20 is good.
\end{lemma}
\begin{proof}
$\hat{T}$ is clearly a tree. To show that $\hat{T}$ is good, we distinguish three cases as follows.

{\em Case 1:} Both $C_1$ and $C_2$ are cycles. In this case, $w(\hat{T}) = |C| + |C_1| + |C_2| 
- 2$, $b(\hat{T}) = |C| + |C_1| + |C_2| \ge 12$, and $|L(\hat{T})| = 2$. Thus, $\hat{T}$ is clearly good. 

{\em Case 2:} One of $C_1$ and $C_2$ is a cycle. W.l.o.g., we assume that $C_2$ is 
a cycle. If $C_1$ is good, then $w(\hat{T}) \ge w(C_1) + |C| + |C_2| - 1 \ge \frac{4}{5}b(C_1)
+ |C| + |C_2| - 1$, $b(\hat{T}) = b(C_1) + |C| + |C_2| \ge b(C_1) + 8 \ge 12$, and 
$|L(\hat{T})| \le |L(C_1)|+1$, implying that $\hat{T}$ is good. So, assume that $C_1$ is bad. Then, 
by Lemma~\ref{lem:type2}, $C_1$ is a 0-path or 4-path whose endpoints are leaves of $G$. 
Indeed, $C_1$ is not a 0-path, because Operation~13 cannot be performed on ${\cal C}_0$. 
Thus, $w(\hat{T}) = |C| + |C_2| + 2$, $b(\hat{T}) = |C| + |C_2| + 4 \ge 12$, and $|L(\hat{T})| = 3$. 
Hence, $\hat{T}$ is good. 

{\em Case 3:} Neither $C_1$ nor $C_2$ is a cycle. If both $C_1$ and $C_2$ are good, 
then $w(\hat{T}) \ge |C| + w(C_1) + w(C_2) \ge |C| + \frac{4}{5}b(C_1) + \frac{4}{5}b(C_2)$, 
$b(\hat{T}) = |C| + b(C_1) + b(C_2) \ge b(C_1) + b(C_2) + 4 \ge 12$, and $|L(\hat{T})| \le |L(C_1)| 
+ |L(C_2)|$, implying that $\hat{T}$ is good. Similarly, if both $C_1$ and $C_2$ are bad, 
then both of them are 4-paths whose endpoints are leaves of $G$ and in turn $w(\hat{T}) = 
|C| + 6$, $b(\hat{T}) = |C| + 8 \ge 12$, and $|L(\hat{T})| = 4$, implying that $\hat{T}$ is good. So, 
we may assume that $C_1$ is good but $C_2$ is bad. Then, $C_2$ is a 4-path whose 
endpoints are leaves of $G$. Hence, $w(\hat{T}) \ge |C| + w(C_1) + 3 \ge 
|C| + \frac{4}{5}b(C_1) + 3$, $b(\hat{T}) = |C| + b(C_1) + 4 \ge b(C_1) + 8\ge 12$, and 
$|L(\hat{T})| \le |L(C_1)| + 2$. Therefore, $\hat{T}$ is good. 
\end{proof}

\begin{lemma}\label{lem:op21}
Operation 21 is good.
\end{lemma}
\begin{proof}
By Lemma~\ref{lem:op19}, Operation~21 is clearly good. 
\end{proof}

\begin{lemma}\label{lem:op22}
Operation 22 is good.
\end{lemma}
\begin{proof}
The operation clearly decreases the number of leaves in $C$ by~1, and is hence good. 
\end{proof}

\begin{lemma}\label{lem:op23}
Operation 23 is good.
\end{lemma}
\begin{proof}
$\hat{T}$ is clearly a tree. To show that $\hat{T}$ is good, we distinguish three cases as follows.

{\em Case 1:} $C_2$ is a cycle. In this case, $w(\hat{T}) = |C_2| + 3$, 
$b(\hat{T}) = |C_2| + 4 \ge 8$, and $|L(\hat{T})| = 3$. So, $\hat{T}$ is clearly good. 

{\em Case 2:} $C_2$ is good. In this case, $w(\hat{T}) \ge w(C_2) + 4 \ge 
\frac{4}{5}b(C_2) + 4$, $b(\hat{T}) = b(C_2) + 4 \ge 8$, and 
$|L(\hat{T})| \le |L(C_2)| + 2$. So, $\hat{T}$ is clearly good. 

{\em Case 3:} $C_2$ is bad. In this case, $C_2$ is either a 0-path or a 4-path whose endpoints 
are leaves of $G$. In the former case, $w(\hat{T}) = 4$, $b(\hat{T}) = 4$, and $|L(\hat{T})| = 3$, 
implying that $\hat{T}$ is clearly good. In the latter case, $w(\hat{T}) = 7$, $b(\hat{T}) = 8 $, 
and $|L(\hat{T})| = 4$, implying that $\hat{T}$ is clearly good.
\end{proof}

We next show that the above operations lead to a number of useful properties of ${\cal C}_2$. 

\begin{lemma}\label{lem:4-cycle1}
Each 4-cycle of ${\cal C}_2$ is adjacent to at most one other connected component 
of ${\cal C}_2$ in $G$. 
\end{lemma}
\begin{proof}
Let $C$ be a 4-cycle in ${\cal C}_2$. Further let $v_1$, \ldots, $v_4$ be the vertices of 
$C$ and assume that they appear in $C$ clockwise in this order. By Statement~\ref{stat:p6} 
in Lemma~\ref{lem:prep}, $C$ has at least three ports. Without loss of generality, we may 
assume that $v_1$ through $v_3$ are ports of $C$. Since Operation~20 cannot be performed 
on ${\cal C}_2$, there is a unique connected component $C'$ in ${\cal C}_2$ such that 
$N_G(\{v_1,v_2\}) \setminus V(C)$ is a nonempty subset of $V(C')$. For the same reason, 
$N_G(\{v_2,v_3\}) \setminus V(C)$ is a nonempty subset of $V(C')$. Moreover, if $v_4$ is 
also a port of $C$, then for the same reason, $N_G(\{v_3,v_4\}) \setminus V(C)$ is 
a nonempty subset of $V(C')$. Therefore, in any case, $N_G(C) \setminus V(C) \subseteq 
V(C')$ and hence $C$ is adjacent to only $C'$ in $G$. 
\end{proof}

\begin{lemma}\label{lem:4-cycle4-cycle}
No two 4-cycles of ${\cal C}_2$ are adjacent in $G$. 
\end{lemma}
\begin{proof}
For a contradiction, assume that two 4-cycles $C_1$ and $C_2$ of ${\cal C}_2$ are 
adjacent in $G$. Then, by Lemma~\ref{lem:4-cycle1}, $G[V(C_1)\cup V(C_2)]$ is a 
connected component of $G$. However, this is impossible because $G$ is connected 
and $|V(G)| \ge 9$. 
\end{proof}

\begin{lemma}\label{lem:4-cycle4-path}
No 4-cycle $C$ of ${\cal C}_2$ is adjacent to a 4-path component of ${\cal C}_2$ in $G$. 
\end{lemma}
\begin{proof}
For a contradiction, assume that a 4-cycle $C$ of ${\cal C}_2$ is adjacent to a 4-path 
component $P$ of ${\cal C}_2$ in $G$. Let $v_1$, \ldots, $v_4$ be the vertices of $C$ 
and assume that they appear in $C$ clockwise in this order. Let $B$ be the set of all 
$u\in V(G) \setminus V(C)$ such that for some $v_i \in V(C)$, $\{u,v_i\} \in E(G)$. 
Since Operation~4 cannot be performed on $G$, $|B| \ge 2$. 
Moreover, by Lemma~\ref{lem:4-cycle1}, $B \subseteq V(P)$. 

Let $u_1$, \ldots, $u_5$ be the vertices of $P$ and assume that they appear in $P$ in 
this order. Then, $P$ is a dead 4-path component of ${\cal C}_0$. 
So, by Statement~\ref{stat:p4} in Lemma~\ref{lem:prep}, both $u_1$ and $u_5$ are leaves 
of $G$. Thus, $B \subseteq \{u_2,u_3,u_4\}$. Since $|C| = 4$ and $C$ has at least three 
ports (by Statement~\ref{stat:p6} in Lemma~\ref{lem:prep}), there are two consecutive 
edges in $C$ whose endpoints all are ports of $C$. Without loss of generality, 
we assume that $v_1$ through $v_3$ are ports of $C$. 

{\em Case 1:} $\{v_2,u_3\} \in E(G)$. In this case, since both $v_1$ and $v_3$ are ports 
of $C$ and Operation~12 cannot be performed on ${\cal C}_0$, $\{u_2,u_3,u_4\} \cap N_G(v_1) 
= \{u_3\}$ and $\{u_2,u_3,u_4\} \cap N_G(v_3) = \{u_3\}$, and in turn $\{u_2,u_3,u_4\} 
\cap N_G(v_2) = \{u_3\}$ as well. Now, since $|B| \ge 2$, $G$ has an edge $\{v_4,u_j\}$ 
with $j \in \{2, 4\}$. However, we can now see that Operation~12 can be performed on 
${\cal C}_0$, a contradiction. 

{\em Case 2:} $\{v_2,u_3\} \not\in E(G)$. In this case, since $v_2$ is a port of $C$, 
$\{v_2,u_2\} \in E(G)$ or $\{v_2,u_4\} \in E(G)$. So, $\{v_1,u_3\}\not\in E(G)$ and 
$\{v_3,u_3\}\not\in E(G)$, because Operation~12 cannot be performed on ${\cal C}_0$. 
Thus, $N_G(\{v_1,v_2,v_3\}) \cap \{u_2,u_3,u_4\}$ is either $\{u_2,u_4\}$ or $\{u_j\}$ 
for some $j\in\{2,4\}$. In the latter case, since $|B| \ge 2$, $\{v_4,u_3\} \in E(G)$ 
or $\{v_4,u_{6-j}\} \in E(G)$, and hence Operation~10 or~12 can be performed, a 
contradiction. In the former case, if $\{v_4,u_3\} \in E(G)$, then Operation~12 can be 
performed on ${\cal C}_0$, a contradiction; otherwise, $N_G(v_4) \subseteq V(C) \cup 
\{u_2,u_4\}$ and in turn Operation~10 can be performed on $G$, a contradiction. 
\end{proof}

\begin{lemma}\label{lem:4-cycle5-cycle}
No 4-cycle of ${\cal C}_2$ is adjacent to a 5-cycle of ${\cal C}_2$ in $G$. 
\end{lemma}
\begin{proof}
For a contradiction, assume that a 4-cycle $C_1$ of ${\cal C}_2$ is adjacent to 
a 5-cycle $C_2$ of ${\cal C}_2$ in $G$. For each $i \in \{1,2\}$, let $v_{i,1}$, 
\ldots, $v_{i,|C_i|}$ be the vertices of $C_i$ and assume that they appear in $C_i$ 
clockwise in this order. By Statement~\ref{stat:p6} in Lemma~\ref{lem:prep}, $C_1$ 
has at least three ports, and in turn three ports of $C_1$ appear in $C_1$ 
consecutively because $|C_1| = 4$. Without loss of generality, we assume that 
$v_{1,1}$, $v_{1,2}$, and $v_{1,3}$ are ports of $C_1$. 

By Lemma~\ref{lem:4-cycle1}, $N_G(C_1) \setminus V(C_1) \subseteq V(C_2)$. Since 
$|C_1| + |C_2| - 1 = 8 < 9$ and Operation~4 cannot be performed on $G$, $|X| \ge 2$, 
where $X$ is the set of vertices $v_{2,j} \in V(C_2)$ with $N_G(v_{2,j}) \setminus 
(V(C_1) \cup V(C_2)) \ne \emptyset$. We distinguish two cases as follows. 

{\em Case 1:} $X$ has two vertices adjacent in $C_2$. Without loss of generality, 
we assume that $\{v_{2,4},v_{2,5}\} \subseteq X$. Then, since Operation~20 cannot 
be performed on ${\cal C}_2$, $N_G(C_1) \setminus V(C_1) \subseteq \{v_{2,2}\}$. 
So, Operation~4 can be performed on $G$, a contradiction. 

{\em Case 2:} No two vertices of $X$ are adjacent in $C_2$. In this case, since 
$|C_2| = 5$, $|X| = 2$. Without loss of generality, we assume that $X = \{v_{2,1},
v_{2,3}\}$. Then, since Operation~20 cannot be performed on ${\cal C}_2$, 
$N_G(C_1) \setminus V(C_1) \subseteq X$. Indeed, since Operation~4 cannot be 
performed on ${\cal C}_2$, $N_G(C_1) \setminus V(C_1) = X$. So, 
$|N_G(\{v_{1,1},v_{1,2}\}) \setminus V(C_1)| = 1$ and $|N_G(\{v_{1,2},v_{1,3}\}) 
\setminus V(C_1)| = 1$ because Operation~10 cannot be performed on $G$. Thus, 
$|N_G(\{v_{1,1},v_{1,2},v_{1,3}\}) \setminus V(C_1)| = 1$ because $N_G(v_{1,2}) 
\setminus V(C_1) \ne \emptyset$. Now, since $|N_G(C_1) \setminus V(C_1)| = 2$, 
$|N_G(\{v_{1,3},v_{1,4}\}) \setminus V(C_1)| = 2$ and in turn Operation~10 can 
be performed on $G$, a contradiction.
\end{proof}

Based on the above lemmas in this section, we are now ready to prove the next lemma:

\begin{lemma}\label{lem:stage2}
Suppose that $C$ is a connected component of ${\cal C}_2$. 
Then, $C$ is a 4-cycle, 5-cycle, 0-path, 4-path, or good connected component. 
Moreover, the following statements hold: 
\begin{enumerate}
\item\label{stat:final1} 
	If $C$ is a 0-path, then its unique vertex $u$ satisfies that for a single tree 
	component $C'$ of ${\cal C}_2$, each $v \in N_G(u)$ is an internal vertex of $C'$, 
	and $u$ is a leaf of $G$ if $C'$ is bad. 
\item\label{stat:final2} 
	If $C$ is a 4-path component of ${\cal C}_2$, then its endpoints are 
	leaves of $G$ and each internal vertex $u$ of $C$ satisfies that each neighbor 
	of $u$ in $G$ is a leaf of $G$, a vertex of a 5-cycle of ${\cal C}_2$, or an 
	internal vertex of a 4-path component or a good connected component of ${\cal C}$. 
\item\label{stat:final3} 
	If $C$ is a 4-cycle of ${\cal C}_2$, then each vertex $u$ of $C$ satisfies 
	that each neighbor of $u$ in $G$ is an internal vertex of a good connected 
	component of ${\cal C}_2$. 
\item\label{stat:final4} 
	If $C$ is a 5-cycle of ${\cal C}_2$, then each vertex $u$ of $C$ satisfies 
	that each neighbor of $u$ in $G$ is an internal vertex of a 4-path component 
	of ${\cal C}_2$. 
\item\label{stat:final5} 
	If $C$ is a good connected component but not a Hamiltonian path of $G$, 
	then each leaf $u$ of $C$ satisfies that each neighbor of $u$ in $G$ is 
	an internal vertex of $C$. 
\end{enumerate}
\end{lemma}
\begin{proof}
By Lemmas~\ref{lem:type2} through \ref{lem:op23}, $C$ is a cycle of length at least~4, 
0-path, 4-path, or good connected component. Indeed, $C$ cannot be a cycle of length~6 
or more, because otherwise Operation~6 could be performed on ${\cal C}_1$ or 
Operation~$i$ could be performed on ${\cal C}_2$ for some $i\in\{15,16,17\}$. 
We next prove the statements separately as follows. 

{\em Statement~\ref{stat:final1}:} Suppose that $C$ is a 0-path. Let $u$ be the unique 
vertex in $C$. Since Operation~18 cannot be performed on ${\cal C}_2$, $N_G(u) \subseteq 
V(C')$ for some connected component $C'$ of ${\cal C}_2$. If $C'$ is not a connected 
component of ${\cal C}_1$, then by Lemmas~\ref{lem:op15} through~\ref{lem:op23}, $C'$ is 
a good connected component of ${\cal C}_2$ and in turn each $v \in N_G(u)$ is an internal 
vertex of $C'$ (because otherwise Operation~19 could be performed on ${\cal C}_2$). 
So, we may assume that $C'$ is also a connected component of ${\cal C}_1$. Then, by 
Lemma~\ref{lem:type2}, $C'$ is a tree component and each $v \in N_G(u)$ is an internal 
vertex of $C'$. For a contradiction, assume that $C'$ is bad but $u$ is not a leaf of $G$. 
Since $C'$ is a bad tree component of ${\cal C}_1$ with internal vertices, 
Lemma~\ref{lem:type2} ensures that $C'$ is a 4-path. 
Let $u_1$, \ldots,$u_5$ be the vertices of $C'$ and assume that they appear in $C'$ in 
this order. Since $u$ is not a leaf of $G$, Lemma~\ref{lem:type2} ensures that $N_G(u) = 
\{u_2,u_4\}$. Now, Operation~8 or~23 can be performed on $G$, a contradiction. 

{\em Statement~\ref{stat:final2}:} Suppose that $C$ is a 4-path component of ${\cal C}_2$. 
Then, $C$ is also a 4-path component of ${\cal C}_1$, because Operation~$i$ does not 
produce a new bad connected component in ${\cal C}$ for each $i\in\{15,\ldots,23\}$. 
So, by Lemma~\ref{lem:type2}, each endpoint of $C$ is a leaf of $G$. Consider an arbitrary 
internal vertex $u$ of $C$ and an arbitrary neighbor $v$ of $u$ in $G$. 
Since Operation~19 cannot be performed on ${\cal C}_2$, $v$ is 
not a leaf of a good connected component of ${\cal C}_2$. So, if $v$ appears in a good 
connected component $C'$ of ${\cal C}_2$, $v$ must be an internal vertex of $C'$. Moreover, 
by Lemma~\ref{lem:4-cycle4-path}, $v$ cannot appear in a 4-cycle of ${\cal C}_2$. Thus, 
to finish the proof, we may assume that $v$ appears in a bad tree component $C''$ of 
${\cal C}_2$. Now, if $v$ appears in a 0-path component of ${\cal C}_2$, then $v$ is 
a leaf of $G$ by Statement~\ref{stat:final1} in this lemma; otherwise, $v$ cannot be 
an endpoint of another 4-path component of ${\cal C}_2$ because each endpoint of 
a 4-path component of ${\cal C}_2$ is a leaf of $G$.  

{\em Statement~\ref{stat:final3}:} Let $C$ be a 4-cycle of ${\cal C}_2$. $C$ cannot be 
adjacent to a 0-path component of ${\cal C}_2$ in $G$, because Operation~6 cannot be 
performed on ${\cal C}_0$ and neither Stage~2 nor Stage~3 produces a new cycle or 
a new 0-path component in ${\cal C}$. So, by Lemmas~\ref{lem:4-cycle4-cycle} 
through~\ref{lem:4-cycle5-cycle}, each $v \in N_G(C)$ appears in a good connected 
component $C'$ of ${\cal C}_2$. Indeed, $v$ must be an internal vertex of $C'$, 
because Operation~19 cannot be performed on ${\cal C}_2$. 

{\em Statement~\ref{stat:final4}:} Let $C$ be a 5-cycle of ${\cal C}_2$. $C$ cannot be 
adjacent to a 0-path component of ${\cal C}_2$ in $G$, because Operation~6 cannot be 
performed on ${\cal C}_0$ and neither Stage~2 nor Stage~3 produces a new cycle or 
a new 0-path component in ${\cal C}$. Moreover, since neither Operation~15 nor 
Operation~16 can be performed on ${\cal C}_2$, $C$ cannot be adjacent to a 5-cycle or 
a good connected component of ${\cal C}_2$ in $G$. So, by Lemma~\ref{lem:4-cycle5-cycle}, 
each $v \in N_G(C)$ appears in a 4-path component $C'$ of ${\cal C}_2$. Indeed, $v$ must 
be an internal vertex of $C'$, because each endpoint of $C'$ is a leaf of $G$. 

{\em Statement~\ref{stat:final5}:} Supppose that $C$ is a good connected component of 
${\cal C}_2$ but not a Hamiltonian path of $G$. Let $u$ be a leaf of $C$. 
Since Operation~$i$ cannot be performed on ${\cal C}_2$ for 
each $i\in\{19,21,22\}$, each $v \in N_G(u)$ is an internal vertex of $C$. 
\end{proof}

Finally, in the third stage, we complete the transformation of ${\cal C}$ 
into a spanning tree of $G$ by further modifying ${\cal C}$ 
by performing the following steps:
\begin{enumerate}
\item For each cycle $C$ of ${\cal C}$, first select an arbitrary edge 
	$e=\{u,v\} \in E(G)$ such that $u \in V(C)$ and $v \in V(G) \setminus V(C)$, 
	then delete one edge incident to $u$ in $C$, and further add $e$. 
	({\em Comment:} Since no two cycles in ${\cal C}_2$ are adjacent in $G$, 
	$v$ appears in a tree component of ${\cal C}$. Moreover, after this step, 
	${\cal C}$ has only tree components.)

\item Arbitrarily connect the connected components of ${\cal C}$ into a tree 
	by adding some edges of $G$. 
\end{enumerate}

It is easy to see that for each $i\in\{15,\ldots,23\}$, Step~$i$ can be done in 
$O(m)$ time. So, the second stage takes $O(nm)$ time. Since the other two stages 
can be easily done in $O(m)$ time, the refined Step~\ref{step:trans} can be done 
$O(nm)$ time.

\section{Performance Analysis}\label{sec:analysis}
Let $g_2$ (respectively, $g_3$) be the number of internal vertices in connected 
components of ${\cal C}_2$ satisfying Condition~C2 (respectively, C3), $b_2$ 
(respectively, $b_3$) be the total number of edges in ${\cal C}_0$ whose endpoints 
appear in the same connected components of ${\cal C}_2$ satisfying Condition~C2 
(respectively, C3), $c_4$ (respectively, $c_5$) be the number of 4-cycles 
(respectively, 5-cycles) in ${\cal C}_2$, and $p_4$ be the number of 4-path 
components in ${\cal C}_2$. 

\begin{lemma}\label{lem:apx}
Let $T_{apx}$ be the spanning tree of $G$ outputted by the refined algorithm. 
Then, the following hold:
\begin{enumerate}
\item\label{stat:apx}
	$w(T_{apx}) \ge 3c_4 + 4c_5 + 3p_4 + g_2 + g_3 \ge 
	3c_4 + 4c_5 + 3p_4 + \frac{4}{5}b_2 + b_3$. 
\item\label{stat:opt1} 
	$opt(G) \le 4c_4 + 5c_5 + 4p_4 + b_2 + b_3$.
\item\label{stat:opt2} 
	$opt(G) \le 3c_4 + 5c_5 + 3p_4 + 2g_2 + 2g_3$. %if $w(T_{apx}) < \frac{4}{5}opt(G)$. 
\end{enumerate}
\end{lemma}
\begin{proof}
We prove the statements separately as follows. 

{\em Statement~\ref{stat:apx}}: Obvious.

{\em Statement~\ref{stat:opt1}}: Clear from the fact that 
	$opt(G) \le |E({\cal C}_0)| \le 4c_4 + 5c_5 + 4p_4 + b_2 + b_3$. 

{\em Statement~\ref{stat:opt2}}: For convenience, 
let $T'$ be obtained from $T$ by rooting $T$ at an internal vertex, and $T''$ 
be obtained from $T'$ by removing those edges $(u,v)$ such that some 4-cycle of 
${\cal C}_2$ contains both $u$ and $v$. Further let $I'$ (respectively, $I''$) 
be the set of vertices in $T'$ (respectively, $T''$) that have at least one 
child in $T'$ (respectively, $T''$). Also let $J = I' \setminus I''$. Clearly, 
$w(T) = |I'|$. Moreover, for each 4-cycle $C$ of ${\cal C}_2$, $T'$ can contain 
at most three edges between the vertices of $C$. So, $|J| \le 3c_4$. Furthermore, 
Lemma~\ref{lem:stage2} ensures that each vertex of $I'$ other than 
\begin{itemize}
\item the vertices in $J$, 
\item the $3p_4$ internal vertices of 4-path components of ${\cal C}_2$, 
\item the $5c_5$ vertices of 5-cycles, and 
\item the $g_2 + g_3$ internal vertices of good connected components of ${\cal C}_2$ 
\end{itemize}
must have a child in $T''$ that is an internal vertex of a good connected component of 
${\cal C}_2$. So, $w(T) = |I'| \le 3c_4 + (5c_5 + 3p_4 + g_2 + g_3) + (g_2 + g_3)$.  
\end{proof}

\begin{theorem}\label{th:main}
The algorithm achieves an approximation ratio of $\frac{13}{17}$ and runs in 
$O(n^2m) + t(2n,2m)$ time. 
\end{theorem}
\begin{proof}
Let $T_{apx}$ be as in Lemma~\ref{lem:apx}, and $r = w(T_{apx}) / opt(G)$. 
By Lemma~\ref{lem:apx}, $r \ge \max\{r_1,r_2\}$, where 
$r_1 = \frac{3c_4 + 4c_5 + 3p_4 + g_2 + g_3}{4c_4 + 5c_5 + 4p_4 + b_2 + b_3}$ 
and $r_2 = \frac{3c_4 + 4c_5 + 3p_4 + g_2 + g_3}{3c_4 + 5c_5 + 3p_4 + 2g_2 + 2g_3}$. 
Note that $r_1 \ge \min\left\{ \frac{4}{5}, r'_1\right\}$ and 
$r_2 \ge \min\left\{ \frac{4}{5}, r'_2\right\}$, where 
$r'_1 = \frac{3c_4 + 3p_4 + g_2 + g_3}{4c_4 + 4p_4 + b_2 + b_3}$ 
and $r'_2 = \frac{3c_4 + 3p_4 + g_2 + g_3}{3c_4 + 3p_4 + 2g_2 + 2g_3}$. 
So, it suffices to show that $\max\{r'_1,r'_2\} \ge \frac{13}{17}$. 
This is done if $r'_1 \ge \frac{13}{17}$. Thus, we assume that $r'_1 < \frac{13}{17}$. 
Then, $c_4 + p_4 > 17g_2 + 17g_3 - 13b_2 - 13b_3$. Hence, 
$r'_2 > \frac{52g_2+52g_3-39b_2-39b_3}{53g_2+53g_3-39b_2-39b_3} \ge 
\min\left\{ \frac{52g_2-39b_2}{53g_2-39b_2}, \, \frac{52g_3-39b_3}{53g_3-39b_3}
\right\}$. Now, since $g_2 \ge \frac{4}{5}b_2$, $\frac{52g_2-39b_2}{53g_2-39b_2}
\ge \frac{13}{17}$. Moreover, since $g_3 \ge b_3$, $\frac{52g_3-39b_3}{53g_3-39b_3}
\ge \frac{13}{14}$. Therefore, $r'_2 > \frac{13}{17}$. 
The running is clearly as claimed.
\end{proof}

Recall that $t(n,m)=O(n^2m^2)$ \cite{Har84}. So, the algorithm takes $O(n^2m^2)$ time.

\end{document}